\documentclass[journal,twoside,web]{ieeecolor}
\usepackage{generic}
\usepackage{cite}
\usepackage{amsmath,amssymb,amsfonts}
\usepackage{graphicx}
\usepackage{hyperref}
\usepackage{textcomp}

\usepackage[ruled,vlined]{algorithm2e}
\usepackage{bm}
\usepackage{booktabs}
\usepackage{color}
\usepackage{multirow}
\usepackage{makecell}
\usepackage{subfigure}

\newtheorem{theorem}{Theorem}
\newtheorem{lemma}{Lemma}
\newtheorem{definition}{Definition}
\newtheorem{remark}{Remark}
\newtheorem{proposition}{Proposition}
\newtheorem{assumption}{Assumption}
\newtheorem{corollary}{Corollary}

\newcommand{\tabincell}[2]{\begin{tabular}{@{}#1@{}}#2\end{tabular}}  
\usepackage{array}

\def\BibTeX{{\rm B\kern-.05em{\sc i\kern-.025em b}\kern-.08em
    T\kern-.1667em\lower.7ex\hbox{E}\kern-.125emX}}
\begin{document}

\title{Resilient Average Consensus with Adversaries via Distributed Detection and Recovery}
\author{Liwei Yuan, \IEEEmembership{Member, IEEE} and Hideaki Ishii, \IEEEmembership{Fellow, IEEE}
\thanks{This work was supported in the part by JSPS under Grant-in-Aid for 
	Scientific Research Grant No.~22H01508, and in the part by the Yuelushan Center for Industrial Innovation under Grant 2023YCII0102.
	}
\thanks{L. Yuan is with the College of Electrical and Information Engineering, Hunan University, Changsha, 410082, China (e-mail: yuanliwei@hnu.edu.cn). }
\thanks{H. Ishii is with the Department of Information Physics and Computing, The University of Tokyo, Tokyo, 113-8656, Japan (e-mail: hideaki\_ishii@ipc.i.u-tokyo.ac.jp). }
}

\maketitle

\begin{abstract}
We study the problem of resilient average consensus in multi-agent systems where some of the agents are subject to failures or attacks. The objective of resilient average consensus is for non-faulty/normal agents to converge to the average of their initial values despite the erroneous effects from malicious agents. To this end, we propose a successful distributed \textit{iterative} resilient average consensus algorithm for the multi-agent networks with general \textit{directed} topologies.
The proposed algorithm has two parts at each iteration: detection and averaging. For the detection part, we propose two distributed algorithms and one of them can detect malicious agents with only the information from direct in-neighbors.
For the averaging part, we extend the applicability of an existing averaging algorithm where normal agents can remove the effects from malicious agents so far, after they are detected.
Another important feature of our method is that it can handle the case where malicious agents are neighboring and collaborating with each other to mislead the normal ones from averaging. This case cannot be solved by existing detection approaches in related literature. Moreover, our algorithm is efficient in storage usage especially for large-scale networks as each agent only requires the values of neighbors within two hops. Lastly, numerical examples are given to verify the efficacy of the proposed algorithms.
\end{abstract}


\section{Introduction}
\IEEEPARstart{D}{istributed}
consensus in multi-agent systems is a fundamental and well-studied topic across different research areas including systems control, computer science, and communication \cite{olfati2007consensus,Lynch, bullo2009distributed,mesbahi2010graph}. 
Under this broad topic, a particular problem that has been extensively studied is that of average consensus where agents try to reach consensus on the average of their values through local interactions among nearby agents \cite{aysal2008distributed,olshevsky2009convergence,cai2011quantized,cai2012average,seyboth2013event,hadjicostis2015robust,chamie2016design,zhu2010discrete,montijano2014robust}. Average consensus algorithms are also useful to maintain the total of the resources invariant and have found applications in, e.g., economic dispatch problems for power systems \cite{yang2013consensus}, distributed computation of PageRank for the search engine of Google \cite{ishii2014pagerank,ishii2018distributed}.
As concerns for cyber-security sharply rise in our society, consensus protocols that properly function even in the presence of faults and adversarial agents have been actively studied; see, e.g., \cite{pasqualetti2012consensus,fagiolini2009dynamic,leblanc2013resilient,yuan2023event,sundaram2018distributed}. The objective is for the non-faulty/normal agents to reach consensus without being affected by the misbehaviors of adversarial agents. In this context, resilient algorithms for performing average consensus have remained somewhat limited except for the recent works \cite{dibaji2019resilient,zheng2021accurate,hadjicostis2022trustworthy}.
A major challenge is that, normal agents should reach consensus on the exact average of their initial values despite adversarial agents' misbehaviors, which may include adding erroneous values to the normal agents' values during the interactions with normal neighbors.

In this paper, we propose an iterative distributed algorithm to tackle the resilient average consensus (RAC) problem in general directed networks under the attacks by the so-called \textit{malicious} agents. Such an agent is capable to send arbitrary but identical values to its neighbors at each iteration \cite{sundaram2011distributed,teixeira2012attack,leblanc2013resilient}. This is the typical way of communication in broadcast networks \cite{goldsmith2005wireless}. 
There are basically two types of approaches for handling the resilient consensus problem, where normal agents need to reach a common value but not necessarily the average of the initial values: (i) mean subsequence reduced (MSR) algorithms \cite{azadmanesh2002asynchronous,vaidya2012iterative,leblanc2013resilient,su2017reaching,dibaji2018resilient,yuan2021resilient,yuan2022asynchronous,senejohnny2019resilience,ishii2022overview} and (ii) detection and isolation algorithms \cite{zhao2018resilient,yuan2021secure}. In MSR algoirthms, agents utilize only the values in a time-varying safety interval to update their next values, with no capabilities to recognize whether a neighbor is adversarial or not.
On the other hand, in detection and isolation algorithms, agents detect the neighbors violating the given consensus protocol and remove the values of such neighbors for updating their next values. This property makes the detection approach a good basis for our RAC algorithm. The reason is that the information of identities of normal agents must be known by the algorithms, which is the key to accumulate the values of normal agents for averaging.

\begin{table}[t]
	\begin{center}
		\caption{Comparisons with Related Resilient Averaging Works}
		\setlength{\tabcolsep}{2mm}{
			\begin{tabular}{lllll} 
				\toprule

				& Algorithm 3 & \cite{hadjicostis2023trustworthy} & \cite{zheng2021accurate} & \cite{dibaji2019resilient}  \\[0.5ex]  
				\midrule
				
				\tabincell{l}{Network\\type}&  Directed& Undirected& Undirected& Directed\\

				\midrule
				
				\tabincell{l}{Adversary\\type} &  Malicious & Malicious& Malicious&  Byzantine\\
				
				\midrule
				
				\tabincell{l}{Neighboring \\adversaries} &  Yes & No & No & Yes\\
				
				\midrule
				
				\tabincell{l}{Communication \\range} &  Two-hop & Two-hop & Two-hop & Flooding\\
				
				\bottomrule
		\end{tabular}}
		\label{table1}
	\end{center}
	\vspace{-7mm}
\end{table}

Our RAC algorithm is based on the detection approach and has two parts: detection and averaging. Existing related works for the RAC problem share this structure \cite{dibaji2019resilient,zheng2021accurate,hadjicostis2023trustworthy}. However, our method has certain advantages over them in different aspects as listed in Table I. More specifically, the work \cite{dibaji2019resilient} proposed a secure broadcast and retrieval algorithm for the RAC problem in directed networks. There, each normal agent uses a certified propagation algorithm to broadcast its initial value to all agents and retrieve the initial values of normal agents for averaging. This approach would cost a huge amount of storage and time for collecting the values of all normal agents in a large-scale network. 
The work \cite{zheng2021accurate} proposed a detection and compensation algorithm for the RAC problem in undirected networks. It utilizes the two-hop neighbors' information to detect misbehaving neighbors and it requires a doubly stochastic adjacent matrix for averaging. As a result, their algorithm is applied in undirected networks only and also cannot handle the case where malicious agents are neighboring with each other. Recently, the authors of \cite{hadjicostis2022trustworthy} proposed an RAC algorithm for directed networks. It allows normal agents to dynamically remove or add the values received from neighbors, however, with the assumption that each normal agent can have access to a correct detection of neighbors. Then in \cite{hadjicostis2023identification,hadjicostis2023trustworthy}, the same authors proposed a detection and compensation algorithm for RAC problem in undirected networks. However, their detection requires the direct communication with two-hop neighbors and it cannot handle the case of neighboring malicious agents either.

In \cite{yuan2021secure}, we proposed a secure detection algorithm for resilient consensus, where each normal agent acts as a detector of its neighbors. An important feature is that it can guarantee the fully distributed detection of malicious neighbors in general directed networks. Besides, it is able to tackle the case of neighboring malicious agents.
This is accomplished through the majority voting \cite{blahut1983theory,parhami1994voting} under a certain topology requirement on the network. In this paper, we exploit these properties and develop a novel RAC algorithm based on the two-hop detection approach.

The contributions of this paper are summarized as follows. We propose a novel RAC algorithm under which normal agents can iteratively detect malicious neighbors and converge to the average of their initial values in general directed networks. 
The proposed algorithm consists of the detection part and the averaging part.
Specifically, for the averaging part, we employ the running-sum based algorithm from \cite{hadjicostis2022trustworthy}, where each node has local buffers to store the total effects received from its in-neighbors. It allows the normal nodes to precisely recover from the influence of malicious neighbors once any misbehavior is detected.
We also improve the applicability of the averaging algorithm by relaxing the necessary assumptions in \cite{hadjicostis2022trustworthy}. In particular, it is sufficient for normal agents to access the correct detection of only in- and out-neighbors for our RAC algorithm, which can save storage resources. Furthermore, we extend the class of misbehaviors of the malicious nodes and consider scenarios where they may go beyond manipulating their identities and also remain to act normally at all times.

For the detection part, we propose two novel algorithms which allow normal nodes to monitor their neighbors and detect as soon as malicious agents perform any misbehaviors in the messages that they broadcast. The fundamental idea is to exploit the two-hop communication so that the normal agents have access to the inputs of their neighbors. This will enable them to obtain multiple reconstructed versions of the outputs of their neighbors and then to compare them to find the true outputs. The difference between the two algorithms lies in the levels in the capabilities for the normal agents to share the detection information among themselves. The first algorithm assumes the availability of authenticated mobile detectors, which help to reduce the requirement on the network connectivity. It will be referred to as the sharing detection algorithm.
Our second algorithm is more significant in that it can be implemented in a fully distributed fashion in our RAC algorithm. Here, each normal node is able to acquire all the inputs of an in-neighbor through the majority voting under a necessary graph structure. Besides, it obtains the detection information of any two-hop in-neighbor (in-neighbor's in-neighbor) by the same approach. As a result, normal nodes can independently detect all the malicious neighbors violating the given averaging algorithm in general \textit{directed} networks. 

Both detection algorithms can handle the case of neighboring malicious nodes, which cannot be solved by related works for the RAC problem \cite{zheng2021accurate,hadjicostis2023identification,hadjicostis2023trustworthy}. 
Moreover, we provide tight graph conditions for our algorithms to achieve the detection and averaging functions, respectively.
We also prove that the graph condition for the fully distributed detection algorithm can be simplified for undirected networks, which makes it more convenient to check whether a graph meets the condition or not.
We emphasize that although the topology requirement may be dense, we can generate the directed/undirected network topologies that satisfy our conditions in large scale.
Lastly, we provide extensive examples to show the efficacy of our RAC algorithm in large-scale networks as well as in an extreme adversarial situation, where over half of the nodes in the network are compromised by malicious attackers.

The rest of this paper is organized as follows. 
Section~II outlines preliminaries on graph notions and the problem settings. Section~III presents the novel RAC algorithm with an emphasis on the averaging part.
Sections~IV and V present the sharing detection algorithm and the fully distributed detection algorithm, respectively. Moreover, tight graph conditions for the proposed algorithms to achieve resilient average consensus are proved.
Section~VI provides numerical examples to demonstrate the efficacy of the proposed algorithms.
Finally, Section~VII concludes the paper.

\section{Preliminaries and Problem Setting}

In this section, we present preliminaries on graph theory, the average consensus algorithm, and the problem settings.

\subsection{Graph Notions}
Consider the directed graph $\mathcal{G} = (\mathcal{V},\mathcal{E})$ consisting of the node set $\mathcal{V}=\{1,...,n\}$ and the edge set $\mathcal{E}\subset \mathcal{V} \times \mathcal{V}$. Here, the edge $(j,i)\in \mathcal{E}$ indicates that node $i$ can receive information from node $j$. Node $j$ is said to be an in-neighbor of node $i$, and node $i$ is an out-neighbor of node $j$.
The sets of in-neighbors and out-neighbors of node $i$ are denoted by $\mathcal{N}_i^-=\{j\in \mathcal{V}:\, (j,i)\in \mathcal{E} \}$ 
and $\mathcal{N}_i^+=\{j\in \mathcal{V}:\, (i,j)\in \mathcal{E} \}$, respectively. 
The in-degree and out-degree of node $i$ are given by $d_i^-=\left| \mathcal{N}_i^-\right| $ and $d_i^+=\left| \mathcal{N}_i^+\right| $, respectively. Here, $\left| \mathcal{S}\right| $ is the cardinality of a finite set $\mathcal{S}$.
If the graph $\mathcal{G}$ is undirected, the edge $(j,i)\in \mathcal{E}$ indicates $(i,j)\in \mathcal{E}$.
A complete graph $\mathcal{K}_n= (\mathcal{V},\mathcal{E})$ is defined by $\mathcal{E} = \{(i, j)\in \mathcal{V} \times \mathcal{V} : i \neq j\}$.
A path from node $i_1$ to $i_m$ is a sequence of distinct nodes $(i_1, i_2, \dots, i_m)$, where $(i_j, i_{j+1})\in \mathcal{E} $ for $j=1, \dots, m-1$. It is also referred to as an $(m-1)$-hop path. We say that node $i_m$ is reachable from node $i_1$. Node $i_1$ is an $(m-1)$-hop in-neighbor of node $i_m$.
A directed graph $\mathcal{G}$ is said to be strongly connected\footnote[1]{An undirected graph is simply said to be connected if every node is reachable from every other node.} if every node is reachable from every other node.
An undirected graph $\mathcal{G}$ is said to be $k'$-connected if it contains at least $k'+1$ nodes and does not contain a set of $k'-1$ nodes whose removal disconnects the graph.

\subsection{Average Consensus and the Running-sum Algorithm}\label{avg_update}

The problem of multi-agent average consensus can be described as follows: Consider a system with $n$ agents interacting over the network modeled by the directed graph $\mathcal{G} = (\mathcal{V},\mathcal{E})$. Each agent $i \in \mathcal{V}$ has a scalar state $x_i[k] \in \mathbb{R}$ to be updated over time $k \in \mathbb{Z}_{\geq 0}$. The goal is to design distributed algorithms that allow agents to eventually converge to the average value of their initial states $\overline{X}= \frac{1}{n} \sum_{i=1}^{n} x_i[0]$, where each agent utilizes only the local information from their neighboring agents during the consensus forming. The push-sum ratio consensus algorithm \cite{kempe2003gossip} was proposed to achieve this goal through two iterative processes on each agent. Here, we describe this algorithm for the time-varying graph $\mathcal{G}[k] = (\mathcal{V},\mathcal{E}[k])$, where $\mathcal{E}[k]\subseteq \mathcal{E}$. Denote the set of out-neighbors of agent $i$ at time $k$ by $\mathcal{N}_i^+[k]$ and the out-degree by $d_i^+[k]= |\mathcal{N}_i^+[k]|$; we employ similar notations for the set of in-neighbors $\mathcal{N}_i^-[k]$ and the in-degree $d_i^-[k]$.

We first introduce the push-sum algorithm, which is the basis of the running-sum algorithm.
Each node $i$ has two state variables, $y_i [k]$ and $z_i [k]$, and updates them as 
\begin{equation}
	\begin{aligned}
		y_i[k+1] \medspace\medspace &= \sum_{j\in \mathcal{N}_i^-[k] \cup \{i\}}  \frac{y_j[k]}{1+d_j^+[k]},  \\
		z_i[k+1] \medspace\medspace &= \sum_{j\in \mathcal{N}_i^-[k] \cup \{i\}}  \frac{z_j[k]}{1+d_j^+[k]},  
	\end{aligned}
	\label{regular_update}
\end{equation}
where $y_i[0]=x_i[0]$ and $z_i[0]=1$ for $i\in \mathcal{V}$. The algorithm
requires each node $i$ to know its out-degree $d_i^+[k]$, and transmit to each out-neighbor the values
\begin{equation}
	\overline{y}_i[k] \medspace\medspace := \frac{y_i[k]}{1+d_i^+[k]}, \medspace\medspace
	\overline{z}_i[k] \medspace\medspace := \frac{z_i[k]}{1+d_i^+[k]} .  
	\label{y_bar}
\end{equation}
Then, by \eqref{regular_update}, these out-neighbors take the sum of received values as their new values.

At each time $k$, node $i$ calculates the ratio 
\[r_i [k] := \frac{y_i [k]}{z_i [k] }. \]
Under some joint connectivity assumptions on the union of the underlying graphs in a certain time window, it was reported in, e.g., \cite{hadjicostis2022trustworthy} that $r_i [k]$ asymptotically converges to the average of the initial values, i.e.,
\begin{equation}
	\lim_{k\to \infty} r_i [k]= \frac{\sum_j y_j[0]}{\sum_j z_j[0]} = \overline{X}, \medspace\medspace
	\forall i \in \mathcal{V}.
\end{equation}

Now, the running-sum ratio consensus algorithm is a variation of the push-sum algorithm used to overcome packet drops or unknown out degrees \cite{hadjicostis2015robust}. It can be summarized as follows. At each time $k$, node $i$ does not send $\overline{y}_i[k]$, $\overline{z}_i[k]$ in \eqref{y_bar} to its out-neighbors. Instead, it sends the so-called $y$ and $z$ running sums denoted by $\lambda$ and $\gamma$, respectively. The two values contain the information of $\overline{y}_i[k]$ and $\overline{z}_i[k]$, and are defined as 
\begin{equation}
	\lambda_i[k+1] \medspace\medspace := \sum\limits_{t=0}\limits^{k}{ \overline{y}_i[t]  },  \medspace\medspace
	\gamma_i[k+1] \medspace\medspace := \sum\limits_{t=0}\limits^{k}{  \overline{z}_i[t] }.  
	\label{lambda_gamma}
\end{equation}
Therefore, an out-neighbor obtains node $i$'s values $\overline{y}_i[k]$, $\overline{z}_i[k]$ by taking the difference of two consecutive $\lambda_i[k]$, $\gamma_i[k]$ as 
\begin{equation*}
	\begin{aligned}
		\overline{y}_i[k] \medspace\medspace &=  \lambda_i[k+1] -\lambda_i[k] ,  \\
		\overline{z}_i[k] \medspace\medspace &=  \gamma_i[k+1] - \gamma_i[k] .  
	\end{aligned}
\end{equation*}
Thus, the running-sum algorithm can achieve average consensus as the push-sum algorithm does, with additional bookkeeping procedures at each node.

Next, we formally outline the structure of the running-sum ratio consensus algorithm \cite{hadjicostis2015robust}. At each time $k$, node $i$ maintains two kinds of values: (i) the running-sum values $\lambda_i[k]$ and $\gamma_i[k]$ of its own; and (ii) the two incoming running-sums from each in-neighbor $j$. More specifically, node $i$ uses $\delta_{ij}[k]$ and $\omega_{ij}[k]$ to keep track of the $y$ and $z$ running sums from node $j$, respectively. They are given as
\begin{equation}
	\begin{aligned}
		\delta_{ij}[k]  &=\lambda_j[k], \medspace\medspace \delta_{ij}[0]=0 ,  \\
		\omega_{ij}[k]  &=\gamma_j[k] , \medspace\medspace \omega_{ij}[0]=0 .  
	\end{aligned}
	\label{delta_omega}
\end{equation}

\subsection{Update Rule and Threat Model}\label{problemsetting}

We now introduce the model of the adversaries and the general structure of the proposed resilient algorithm. First, the node set $\mathcal{V}$ is partitioned into the set of normal nodes $\mathcal{N}$ and the set of adversary nodes $\mathcal{A}$. The latter set is unknown to the normal nodes at time $k=0$. The adversary nodes in $\mathcal{A}$ try to prevent the normal nodes in $\mathcal{N}$ from reaching average consensus. All algorithms in this paper are synchronous.

In our problem setting, the adversary nodes can be quite powerful. We assume that they may behave arbitrarily, deviating from the protocols with which the normal nodes are equipped. Here, we define the threat model of this paper; see also \cite{leblanc2013resilient, dibaji2019resilient, yuan2021secure,hadjicostis2023identification}.

\begin{definition}
	\textit{($f$-total / $f$-local set)}
	The set of adversary nodes $\mathcal{A}$ is said to be $f$-total
	if it contains at most $f$ nodes, i.e., $\left| \mathcal{A}\right| \leq f$.
	Similarly, it is said to be $f$-local
	if for any normal node $i\in \mathcal{N}$, it has at most $f$ adversary in-neighbors, i.e., $\left|\mathcal{N}_i^- \cap \mathcal{A}\right| \leq f, \forall i \in \mathcal{N}$.
\end{definition}

\begin{definition}
	\textit{(Malicious nodes)}
	An adversary node $i\in \mathcal{A}$ is said to be malicious if it changes its own value arbitrarily and sends the same value\footnote[2]{It may also decide not to make a transmission at any time. This corresponds to the crash model \cite{Lynch}.} to its neighbors at each transmission. 
\end{definition}

In this paper, we focus on the malicious model. This model is reasonable in applications such as wireless sensor networks and robotic networks, where neighbors' information is obtained by broadcast communication or vision sensors \cite{goldsmith2005wireless}. This model is different from the Byzantine model, which is well-studied in computer science \cite{Lynch}. Specifically, a Byzantine node can send different values to its different neighbors.
Here, we define a connectivity notion for directed graphs. A directed graph $\mathcal{G}$ is said to be $k'$-strongly connected if after removing any set of nodes satisfying the ($k'-1$)-local model, the remaining digraph is strongly connected.

As mentioned in the Introduction, the proposed algorithm for resilient average consensus is based on detection of the malicious nodes in the network. To this end, each normal node $i$ is equipped with a detection algorithm to monitor the behaviors of its own neighbors. The output of such an algorithm will be the set of malicious nodes known or detected by node $i$ by time $k$ and is denoted by $\mathcal{A}_i[k]$. 

The overall structure of the proposed algorithm is as follows. At each time $k$, each normal node $i$ forms an information set denoted by $\Phi_i[k]$. This set will be shared with its out-neighbors, who will make use of it for their averaging and detection algorithms. The exact contents of the information sets  will be given in the next subsection. Specifically, node $i$ conducts the four steps given below at time $k+1$:

\textit{1.~Transmit} the information set $\Phi_i[k]$ (described in \eqref{phi} later) and the detection information $\mathcal{A}_i[k]$ to all its out-neighbors $j\in \mathcal{N}_i^+$.

\textit{2.~Receive} the information sets $\Phi_j[k]$ and the detection information $\mathcal{A}_j[k]$ from all in-neighbors $j\in \mathcal{N}_i^-$.

\textit{3.~Detect} neighbors according to the detection algorithm to obtain $\mathcal{A}_i[k+1]$.

\textit{4.~Update} $x_i[k+1]$ according to the resilient average consensus algorithm.

The RAC algorithm in Step 4 will be outlined in Section~\ref{sec_rac} whereas the detection algorithm in Step 3 will be given in Sections~\ref{Secfors1} and \ref{Secfors2}.

\subsection{Detection of Adversaries and Information Sets}\label{infoset}

We now describe the general approach for our detection algorithms, based on the ideas from \cite{yuan2021secure}. As mentioned above, each normal node monitors its neighboring nodes and checks if any inconsistencies can be found in their behaviors. In particular, our approach employs two-hop communication among the nodes. That is, each node sends the information received from its direct in-neighbors to its out-neighbors together with its own information. We assume that each node receives information from its two-hop in-neighbors via a sufficient number of different paths. Then, if any of its direct in-neigbors make changes in the information to be passed on, there will be inconsistencies in the data, which can lead to detections of misbehaviors.
To formalize this approach, in this subsection, we first introduce the key notion of information sets of the nodes and then provide assumptions regarding these sets for both normal and malicious nodes. 

Information sets define the data exchanged within the network for performing detection and averaging.
Node $i$'s information set $\Phi_i[k]$ to be broadcasted at time $k+1$ is 
\begin{equation}
	\begin{aligned}
		\Phi_i[k]= & \Big( \mathcal{A}_i[k], (i,\delta_{ii}[k+1|k], \omega_{ii}[k+1|k]),   \\ 
		&   \{(j,\delta_{ij}[k|k], \omega_{ij}[k|k])\}_{j\in \mathcal{N}_i^-\cup\{i\}} \Big).
	\end{aligned}
	\label{phi}
\end{equation}
It has three parts. The first is the set $\mathcal{A}_i[k]$ of adversaries detected by node $i$ by time $k$. The second and the third are node $i$'s own and its in-neighbors' information.
We use the notation $\delta_{ii}[k+1|k]$ to indicate that this value is in the set $\Phi_i[k]$ from time $k$. Note that $\Phi_i[k-1]$ and $\Phi_i[k]$ contain $\delta_{ii}[k|k-1]$ and $\delta_{ii}[k|k]$, respectively, and if node $i$ is malicious, these values may be different.

Next, we introduce assumptions on the nodes' knowledge and the attacks generated by the malicious nodes.

\begin{assumption}\label{neighborinfo}
	Each node $i\in \mathcal{N}$ has access to the information sets received from its in-neighbors. 
	It knows the indices and topology of its two-hop in-neighbors and those of its direct out-neighbors. 
\end{assumption}

\begin{assumption}\label{cannotadd}
	Each node $i\in \mathcal{A}$ may have all the information of the entire network including the topology and state values of all nodes and may cooperate with other malicious nodes even if no edges exist. It can manipulate its own information set in \eqref{phi} and broadcasts the same set to out-neighbors.
\end{assumption}

By Assumption~\ref{neighborinfo}, each normal node has only partial knowledge about the network. To perform detection based on two-hop communication, normal nodes are aware of the topology of two-hop in-neighbors. This setting may be justified in sensor networks when the nodes are geographically fixed and the network topology remain the same. Similar settings are studied in \cite{zheng2021accurate,hadjicostis2023identification,hadjicostis2023trustworthy}. We should highlight that this assumption can be met relatively easily and is of low cost. In MSR-based resilient consensus algorithms \cite{dibaji2018resilient,leblanc2013resilient}, it is sufficient that fault-free nodes have access to the information only from their one-hop neighbors. Clearly, this requirement is weaker than Assumption~\ref{neighborinfo}, but MSR-based algorithms are not capable to detect malicious agents (though they can avoid their influences). Also, in contrast, each fault-free node must know the topology of the entire network in related works based on observer-based detection \cite{pasqualetti2012consensus}, \cite{sundaram2011distributed}, multi-hop communication \cite{Lynch}, \cite{su2017reaching}, and Byzantine agreement \cite{tseng2015fault}.

On the other hand, in Assumption~\ref{cannotadd}, since two-hop communication is employed, a malicious node may modify not only its own states but also those received from its in-neighbors, which are part of its information set. This means that there are more options in terms of attacks compared to, e.g., the MSR-based algorithms. However, we emphasize that such attacks can be detected. For example, a malicious node may add or delete some pairs of agent IDs and values in its information set. It may also decide to remove information from some of its in-neighbors. Since the normal agents have the knowledge of up to their two-hop in-neighbors, attacks will be found by their direct out-neighbors. Moreover, in the case that a malicious node adopts an ID of another node, such attacks can be detected too \cite{yuan2021secure}.
Therefore, our approach does not assume that each node should identify the senders of incoming messages, which is imposed in \cite{hadjicostis2022trustworthy,hadjicostis2023trustworthy}.

\section{Resilient Average Consensus}\label{sec_rac}

In this section, we define the RAC problem and introduce our algorithm with an emphasis on the averaging part.

\subsection{Problem Statement}

Consider a time-invariant directed graph $\mathcal{G} = (\mathcal{V},\mathcal{E})$.
In our \textit{resilient average consensus (RAC)} problem, the goal is to design distributed algorithms that allow normal agents to eventually converge to the average value of their initial states, i.e.,
\begin{equation}\label{RAC}
	x_i[k] \to \overline{X}_{\mathcal{N}_0} := \frac{\sum_{i\in\mathcal{N}_0} x_i[0]} {|\mathcal{N}_0|} \medspace\medspace \textup{as} \medspace\medspace   k \to \infty, \medspace\medspace \forall i \in \mathcal{N} ,
\end{equation}
regardless of the adversarial actions taken by the nodes in $\mathcal{A}$.
In \eqref{RAC}, the resilient average consensus is not defined on the true set $\mathcal{N}$ of normal agents but on the set $\mathcal{N}_0  \supseteq \mathcal{N}$ of all nodes that behave properly over time. This is justified for the case where an adversary node acts as normal for all times. In this case, such an adversary agent's value is included in the average computing since there is no way to detect such adversary nodes. See Section~\ref{discuss22} for more discussions.

Similar problems have been studied in related works \cite{dibaji2019resilient,zheng2021accurate,hadjicostis2023trustworthy}. However, our approach has advantages over these works in different aspects as we discuss in due course.

\subsection{Overview of the RAC Algorithm}

To solve the RAC problem in a distributed iterative fashion, normal agents must know whether their neighbors are malicious or normal. Thereafter, they only interact with normal ones to obtain the desired average. Hence, our RAC algorithm contains two parts: (i) detection and (ii) averaging. The detection algorithm guarantees that each normal node can detect any malicious in-/out-neighbors. On the other hand, the averaging algorithm needs to ensure that each normal node can remove the erroneous effects received from malicious neighbors by the time those neighbors are detected as malicious.

In this section, we first present the averaging algorithm based on the RAC approach of \cite{hadjicostis2022trustworthy}, where each normal node is assumed to have access to the correct detection information of all normal nodes in the network. Note that the detection approach is not discussed in \cite{hadjicostis2022trustworthy}.
For ease of presentation, we assume that every normal node can obtain the correct detection of malicious neighbors by a certain time $k_c$.
Our detection algorithms presented later in Sections~\ref{Secfors1} and \ref{Secfors2} are tailored for working with this averaging algorithm and realize the important function of correct detection.

Recall that in the running-sum algorithm, each agent maintains two variables $\lambda$ and $\gamma$ to record the sum of its own $y$ and $z$ values from the initial time. This feature makes it a good basis for our RAC algorithm. Moreover, for the running-sum algorithm to achieve average consensus, the adjacency matrix needs to be column stochastic, which is easy to realize in directed networks. In contrast, the related resilient averaging works for undirected networks \cite{zheng2021accurate, hadjicostis2023trustworthy} are based on average consensus via linear iterations \cite{olfati2007consensus}, which require the adjacency matrix to be doubly stochastic. However, it may be difficult to design such an adjacency matrix for directed networks, and even tougher for time-varying networks.

Next, we introduce the major steps of our RAC algorithm.
At each time $k$, each normal node $i$ utilizes the detection algorithm to update its detection information regarding in-/out-neighbors in $\mathcal{A}_i[k]$. Then it updates the set of non-faulty in-neighbors as $\mathcal{M}_i^-[k]=\mathcal{N}_i^- \setminus \mathcal{A}_i[k] $ and updates the set of non-faulty out-neighbors as $\mathcal{M}_i^+[k]=\mathcal{N}_i^+ \setminus \mathcal{A}_i[k] $.
Simultaneously, the out-degree is updated by $d_i^+[k]= |\mathcal{M}_i^+[k]|= |\mathcal{N}_i^+ \setminus \mathcal{A}_i[k]|$. 
Given the new $\mathcal{A}_i[k] $, node $i$ updates its $y$ and $z$ using only the running sums from the in-neighbors in $\mathcal{M}_i^-[k]\cup \{i\}$:
\begin{equation}
	\begin{aligned}
		y_i[k] \medspace\medspace &= \sum\limits_{j\in \mathcal{M}_i^-[k] \cup \{i\}}{(\delta_{ij}[k]-\delta_{ij}[k-1])},  \\
		z_i[k] \medspace\medspace &= \sum\limits_{j\in \mathcal{M}_i^-[k] \cup \{i\}}{(\omega_{ij}[k]-\omega_{ij}[k-1])}.  
	\end{aligned}
	\label{yzkplus1}
\end{equation}

By the assumption that every normal node obtains the correct detection of malicious neighbors by time $k_c$, eq.~\eqref{yzkplus1} constrains the averaging within only normal nodes for time $k>k_c$. Therefore, the running-sum algorithm on normal nodes achieves ratio consensus of the values of normal agents at time $k_c$ if the subgraph of normal nodes (i.e., the \textit{normal network}) is strongly connected. However, due to possible erroneous effects from malicious neighbors, the sum of values of normal agents at time $k_c$ may not be the sum of initial values of normal agents. 
Thus, if the erroneous effects from malicious neighbors can be subtracted by normal agents and normal agents' values sent to malicious neighbors can be compensated precisely, then normal agents can recover the sum of their initial values and achieve resilient average consensus.

\subsection{Removing Malicious Effects Based on Detection}

In this part, we introduce how the normal agents conduct the subtraction of in-coming malicious values and compensation of out-going normal values, respectively.
The two actions are different for the cases where in-neighbors or out-neighbors are detected as malicious for the first time.
Note that the actions taken by each node $i$ are based on its own detection $\mathcal{A}_i[k]$.

Case 1: A malicious in-neighbor $j$ is detected for the first time at time $k$.
In this case, node $i$ not only ignores node $j$’s values for updating as in \eqref{yzkplus1} but also removes the effects received from node $j$ so far, i.e., $\delta_{ij}[k], \omega_{ij}[k]$ in \eqref{delta_omega}. This subtraction of in-coming malicious values has to be done for each in-neighbor $j$ in the set $\Delta \mathcal{M}_i^-[k]= \mathcal{M}_i^-[k-1] \setminus \mathcal{M}_i^-[k] $, 
which consists of node $i$'s in-neighbors that are detected as malicious at time $k$. Specifically, we replace $y_i[k]$ and $z_i[k]$, respectively, with
\begin{equation}
	\begin{aligned} 
		y_i[k] \medspace\medspace &= \medspace\medspace y_i[k] - \sum\limits_{j\in \Delta\mathcal{M}_i^-[k]}{\delta_{ij}[k-1]},\\
		z_i[k] \medspace\medspace &= \medspace\medspace z_i[k] - \sum\limits_{j\in \Delta\mathcal{M}_i^-[k]}{\omega_{ij}[k-1]}.
	\end{aligned}
	\label{y_change_in}
\end{equation}

Case 2: A malicious out-neighbor $q$ is detected for the first time at time $k$. In this case, node $i$ not only decreases its out-degree by one ($d_i^+[k]=d_i^+[k-1]-1$) but also compensates for all its own values sent to node $q$ while $q$ was considered normal. It does so by adding to its $y$ and $z$ values its own $y$ and $z$ running sums ($\lambda_i[k],\gamma_i[k]$), respectively.
Similar to Case 1, this adjustment has to be done for every out-neighbor that is detected as malicious at time $k$. Let
$\Delta \mathcal{M}_i^+[k]= \mathcal{M}_i^+[k-1] \setminus \mathcal{M}_i^+[k] $ 
be the set of node $i$'s out-neighbors that are detected as malicious at time $k$. Then $y_i[k]$ and $z_i[k]$ are updated as
\begin{equation}
	\begin{aligned} 
		y_i[k] \medspace\medspace &= \medspace\medspace y_i[k] + |\Delta \mathcal{M}_i^+[k]|\lambda_i[k],\\[1mm]
		z_i[k] \medspace\medspace &= \medspace\medspace z_i[k] + |\Delta \mathcal{M}_i^+[k]|\gamma_i[k].
	\end{aligned}
	\label{y_change_out}
\end{equation}
This is needed because, e.g., $\lambda_i[k]$ is the cumulative $y$ values that were sent to any malicious out-neighbor by time $k$.

\subsection{Convergence Analysis}

Finally, we are ready to present our RAC algorithm in Algorithm~1.
By Assumption~\ref{neighborinfo}, node $i$ knows the sets $\mathcal{N}_i^-$, $\mathcal{N}_i^+$. Moreover, through our detection algorithms, node $i$ keeps track of the set of misbehaving in-/out-neighbors $\mathcal{A}_i[k]$ that it detected by time $k$. Then, by following the above two processes for malicious in-/out-neighbors, resilient average consensus can be achieved with Algorithm~1. The following proposition is the main convergence result for this algorithm from \cite{hadjicostis2022trustworthy} with some enhanced applicability.

\begin{proposition}\label{thm_consensus}
	Consider the directed network $\mathcal{G} = (\mathcal{V},\mathcal{E})$, where each node $i\in \mathcal{V}$ has an initial value $x_i[0]$. Under Assumptions~\ref{neighborinfo} and \ref{cannotadd}, if each normal node can detect all the malicious in- and out-neighbors, and the normal network is strongly connected, then the normal nodes executing Algorithm~1 converge to the average of their initial values given by $\overline{X}_{\mathcal{N}_0}= \frac{\sum_{i\in\mathcal{N}_0} x_i[0]} {|\mathcal{N}_0|}$ in \eqref{RAC} as $k \to \infty$.
\end{proposition}
\vspace*{2.0mm}

\begin{remark}
	The convergence results of this proposition have appeared in \cite{hadjicostis2022trustworthy}; hence, we omit the proof for brevity. The key idea is to show that the sums of the $y$ and $z$ values of normal agents remain invariant at all times. Moreover, it is emphasized that we have improved the results for the averaging algorithm by relaxing the required assumptions as well as justifying the case where adversary agents act normally at all times. Specifically, the results in \cite{hadjicostis2022trustworthy} need the assumption that each transmission of any node $i\in\mathcal{V}$ is associated with a unique node ID that allows the receiver to identify the sender. In contrast, we have discussed in Section~\ref{infoset} that we do not assume that each malicious node must send its real ID as the misbehavior of changing ID can be detected by our algorithms. Besides, the work \cite{hadjicostis2022trustworthy} requires the assumption that each normal node knows the correct detection information of all normal nodes in the network after time $k_c$. However, each normal node is supposed to detect only the malicious in- and out-neighbors in Proposition~\ref{thm_consensus}, which can save storage resources. We will show later that such detection can be realized through our fully distributed detection algorithm.
\end{remark}

In the following sections, we present two detection algorithms for our RAC algorithm, which are redesigned based on the two-hop detection approach in our previous work \cite{yuan2021secure}. In Section~\ref{Secfors1}, we propose the sharing detection algorithm for undirected networks. In Section~\ref{Secfors2}, we present the fully distributed detection algorithm for general directed networks.
The latter algorithm is fully distributed and is efficient in storage usage compared to the related works \cite{dibaji2019resilient,hadjicostis2022trustworthy, hadjicostis2023trustworthy}, where each normal agent must obtain the correct detection of all malicious or normal agents in the network.

\begin{algorithm}[]
	\caption{Resilient Average Consensus Algorithm}
	\LinesNumbered 
	\KwIn{Node $i$ knows $x_i[0]$, $\mathcal{N}_i^-$, $\mathcal{N}_i^+$ by Assumption~\ref{neighborinfo}.}
	
	\SetKwBlock{newbox}{Initialization:}{}
	\newbox{
		\SetAlgoVlined
		Node $i$ initializes $\mathcal{A}_i[0]=\mathcal{A}_i[1]= \emptyset$, and $y_i[0]= x_i[0]$, $\lambda_{i}[0]= 0$, $\delta_{ij}[0]= 0$, $\forall j\in \mathcal{N}_i^-$, \
		$z_i[0]= 1$, $\gamma_i[0]= 0$, $\omega_{ij}[0]= 0$, $\forall j\in \mathcal{N}_i^-$.
		
		At $k=1$, send $\lambda_i[1]$, $\gamma_i[1]$ using eq. \eqref{lambda_gamma} to $\forall q\in \mathcal{N}_i^+$ and receive $\delta_{ij}[1]$, $\omega_{ij}[1]$ from $\forall j\in \mathcal{N}_i^-$.
		
		Obtain $y_i[1]$, $z_i[1]$ using eq. \eqref{regular_update}.
		
		Obtain $\lambda_i[2]$, $\gamma_i[2]$ using eq. \eqref{lambda_gamma}.
		
	}
	
	\For{ $k\geq2$ }{
		
		\textbf{Transmit:} $\Phi_i[k-1]$ to $\forall q\in \mathcal{N}_i^+$.
		
		\textbf{Receive:} $\Phi_j[k-1]$ from $\forall j\in \mathcal{N}_i^-$.
		
		\textbf{Detect:} in-/out-neighbors according to the Detection Algorithm to obtain $\mathcal{A}_i[k]$.
		
		\SetKwBlock{newbox}{Update using detection of in-neighbors:}{}
		\newbox{
			\SetAlgoVlined
			Set $\mathcal{M}_i^-[k]=\mathcal{N}_i^- \setminus \mathcal{A}_i[k]$.
			
			In Case 1: 
			
			For each $j\in \mathcal{N}_i^- \cup \{i\}$, set 
			
			$\delta_{ij}[k] = \left\{
			\begin{aligned} 
				&\lambda_j[k], & & \forall j \in \mathcal{M}_i^-[k] \cup \{i\},   \\
				&0,  & &\textup{otherwise}.   
			\end{aligned}
			\right.
			$  
			
			$\omega_{ij}[k] = \left\{
			\begin{aligned} 
				&\gamma_j[k], & & \forall j \in \mathcal{M}_i^-[k] \cup \{i\},   \\
				&0,  & &\textup{otherwise}.   
			\end{aligned}
			\right.
			$  
		}
		
		\SetKwBlock{newbox}{Compute:}{}
		\newbox{ 
			
			$y_i[k] \medspace\medspace = \sum\limits_{j\in \mathcal{N}_i^- \cup \{i\}}{(\delta_{ij}[k]-\delta_{ij}[k-1])}$,  
			
			$z_i[k] \medspace\medspace = \sum\limits_{j\in \mathcal{N}_i^- \cup \{i\}}{(\omega_{ij}[k]-\omega_{ij}[k-1])}$.
		}
		
		\SetKwBlock{newbox}{Update using detection of out-neighbors:}{}
		\newbox{
			\SetAlgoVlined
			Set $\mathcal{M}_i^+[k]=\mathcal{N}_i^+ \setminus \mathcal{A}_i[k]$.
			
			Set $d_i^+[k] = |\mathcal{M}_i^+[k] | $.
			
			Set $\Delta \mathcal{M}_i^+[k]= \mathcal{M}_i^+[k-1] \setminus \mathcal{M}_i^+[k] $.
			
			In Case 2: 
			
			$y_i[k] \medspace\medspace = \medspace\medspace y_i[k] + |\Delta \mathcal{M}_i^+[k]|\lambda_i[k]$,  
			
			$z_i[k] \medspace\medspace = \medspace\medspace z_i[k] + |\Delta \mathcal{M}_i^+[k]|\gamma_i[k]$.
		}
		
		\SetKwBlock{newbox}{Compute:}{}
		\newbox{ 
			
			$\lambda_i[k+1] \medspace\medspace = \medspace\medspace \lambda_i[k] + y_i[k]/(1+d_i^+[k])$,  
			
			$\gamma_i[k+1] \medspace\medspace = \medspace\medspace \gamma_i[k] + z_i[k]/(1+d_i^+[k])$.
		}
		
		\KwOut{$ r_i[k] = y_i[k] / z_i[k] $}
		
	}
	
\end{algorithm}

\section{Sharing Detection in Undirected Networks}\label{Secfors1}

We introduce our first distributed detection algorithm to be presented as Algorithm~2, where the normal nodes are capable to detect malicious neighbors by using the two-hop information in undirected networks. It provides the basis for the two-hop detection in an adversarial environment, motivated by the works \cite{zhao2018resilient, yuan2021secure}.

\subsection{Detection Algorithm Design}

\begin{algorithm}[t]
	\caption{Sharing Detection Algorithm} 
	\LinesNumbered 
	\KwIn{$\Phi_j[k-1], \forall j\in \mathcal{N}_i^- \cup \{i\} $}

	\SetKwBlock{newbox}{Initialization:}{}
	\newbox{
		\SetAlgoVlined
		
		Node $i$ follows the initialization in Algorithm~1.
		
		At $k=1$, let $\mathcal{A}_i[1]$ include the IDs of in-neighbors not sending initial values to $i$.
		
		Let $\mathcal{C}_i[1]=\{\delta_{ij}[1], \omega_{ij}[1], \forall j \in \mathcal{N}_i^- \cup \{i\} \}$ be the initial check set.
	}
	
	\For{ $k\geq 2$ }{
		
		Let $\mathcal{A}_i[k]=\bigcup_{v\in \mathcal{N}} \mathcal{A}_v[k-1]$ by Assumption \ref{broadcast}. \
		
		\For{ $j\in \mathcal{M}_i^-[k]$ }{
			
			(Step 1) \If{$\mathcal{A}_j[k-1]\neq\mathcal{A}_i[k]$   }{
				let $j \in \mathcal{A}_i[k]$.
			}

			(Step 2) \If{any ID in $\Phi_j[k-1]$ $ \notin \mathcal{N}_j^- \cup \{j\} $ }{
				let $j \in \mathcal{A}_i[k]$.
			}
			
			(Step 3) \If{any of $\delta_{jh}[k-1|k-1]$ or $\omega_{jh}[k-1|k-1]$ 
				in $\Phi_j[k-1]$ is not equal to the corresponding value in $\mathcal{C}_i[k-1]$  }{
				let $j \in \mathcal{A}_i[k]$.
			}
			
			(Step 4) \If{any of $\delta_{jj}[k|k-1]$ or $\omega_{jj}[k|k-1]$ in $\Phi_j[k-1]$ is not equal to the reconstructed $\lambda_{j}'[k]$ or $\gamma_{j}'[k]$ by node $i$ }{
				let $j \in \mathcal{A}_i[k]$.
			}
			
		}
		
		\KwOut{$\mathcal{A}_i[k]$} 
		
		\textbf{Store:} $\delta_{jj}[k|k-1] $ and $\omega_{jj}[k|k-1] $ from $ \Phi_j[k-1]$, $ j\in \mathcal{N}_i^- \cup \{i\}$, into $\mathcal{C}_i[k]$.
		
	}
	
\end{algorithm}

For this algorithm, the sharing detection function below is needed for the communication among the nodes when events of detecting adversaries occur.

\begin{assumption}\label{broadcast}
	Once a malicious node is detected by any normal node at any time step, its ID will be securely notified to all nodes within the same time step.
\end{assumption}

This assumption also appeared in \cite{zhao2018resilient,yuan2021secure} for resilient consensus. As we reported in \cite{yuan2021secure}, the sharing detection function can be realized by introducing fault-free mobile nodes which are appropriately distributed throughout the network and are capable to immediately verify if the detection reports from a node is true or false. 
Note that the deployment of such mobile agents is only for verification of detection reports instead of detecting malicious agents by themselves.

The sharing function is crucial for Algorithm~2 since it is necessary for detecting malicious nodes that are neighboring and cooperating with each other. We must emphasize that if we do not have this function for Algorithm~2, it can only handle the case where no neighboring malicious nodes exist (i.e., any malicious node is surrounded by normal nodes only). This is exactly the case studied in related works \cite{zheng2021accurate, hadjicostis2023identification,hadjicostis2023trustworthy}.

We now present Algorithm~2. To ensure that all nodes follow the correct averaging in Algorithm~1, the normal nodes check consistency among the neighbors' information sets. In Algorithm~2, step 1 is to guarantee that each normal node should not use the values from the nodes detected as malicious already. Moreover, it ensures that a node does not falsely claim another node being malicious. Step~2 is to prevent the malicious nodes from faking any neighbors. Step~3 is to enforce the normal nodes not to modify the values received from their neighbors. Finally, step~4 is to guarantee that the neighbors follow the averaging in Algorithm~1.

\subsection{Necessary Graph Structure for Algorithm 2}

\begin{figure}[]
	\centering
	
	\hspace{9pt}
	\subfigure[]{
		\includegraphics[width=1.5cm]{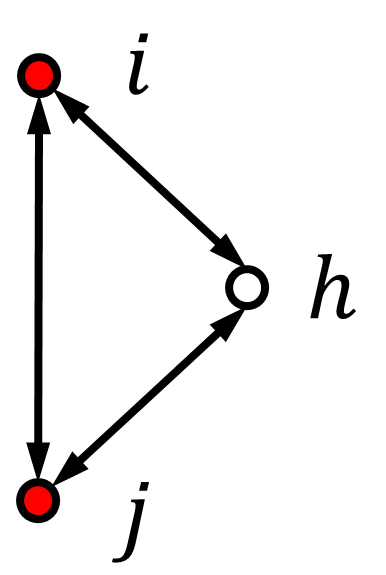}
	}
	\quad
	\hspace{5pt}
	\subfigure[]{
		\includegraphics[width=2.3cm]{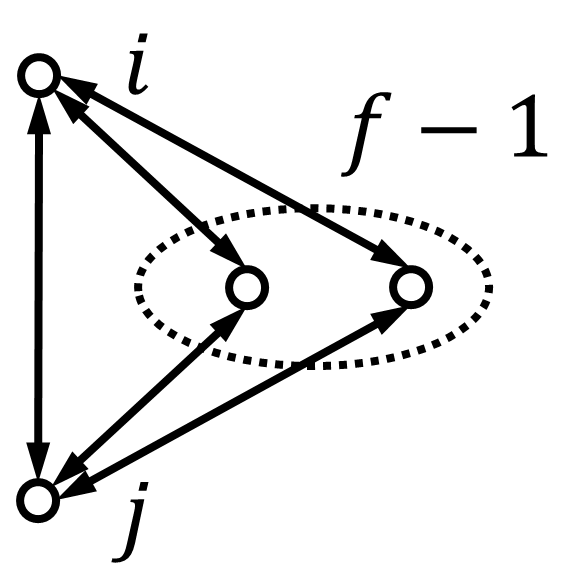}
	}
	
	\vspace{-3pt}
	\caption{Illustration of the graph condition for Algorithm~2: (a) There is at least one common normal neighbor between any pair of neighboring malicious nodes. (b) There are at least $f-1$ common neighbors between any pair of neighbors under the $f$-total model.}
	\label{share_condition}
\end{figure}

In this part, we provide the necessary graph condition for Algorithm~2. We can observe that a malicious node can be detected if there is at least one normal node among its neighbors that monitors its behavior. However, such detection may fail if neighboring malicious nodes cooperate with each other. Hence, it is critical that one or more normal nodes are present as their common neighbors. We illustrate this graph structure in Fig.~\ref{share_condition}(a). Here, nodes~$i$ and~$j$ are malicious. They can cooperate as follows: Node~$i$ manipulates $\delta_{ij}[k|k]$ in its information set, and node~$j$ manipulates $\delta_{ji}[k|k]$ in its information set. If there is no normal node having access to the information sets of both nodes~$i$ and~$j$, such an attack will not be detected. In contrast, the detection works if there is a common neighbor $h$ of nodes~$i$ and~$j$.

Next, we state the necessary and sufficient graph condition for Algorithm~2 to detect all the misbehaving agents.

\begin{lemma}\label{lemma_share} 
	Consider the undirected graph $\mathcal{G}=(\mathcal{V},\mathcal{E})$. 
	Algorithm~2 detects every pair of neighboring misbehaving nodes if and only if they have at least one normal common neighbor.
\end{lemma}

\begin{proof}
	We can show similarly to Lemma~8 in \cite{yuan2021secure} except that the update check in step 4 in Algorithm~2 is more complex than the general consensus protocol used there. Here, we provide a sketch of the proof. In the undirected network using Algorithm~2, each normal node $i$ can at least verify its own values $\delta_{ji}[k|k]$ and $\omega_{ji}[k|k]$ in $\Phi_j[k]$, $j\in \mathcal{N}_i^-$. If a malicious node is only surrounded by normal agents, then it cannot change any $\delta_{ji}[k|k]$ and $\omega_{ji}[k|k]$ values from neighbors. Moreover, normal neighbor $i$ can reconstruct $\lambda_{j}[k+1]$ or $\gamma_{j}[k+1]$ through the averaging part in Algorithm~1 to check if node $j$ is following the averaging or not. Thus, misbehaving node $j$ will be detected by all normal neighbors. In the case of neighboring malicious nodes, they can modify the values from each other, but this is also detected by the normal common neighbor of them as discussed before. 
\end{proof}

Given that the malicious nodes are unknown and possibly cooperate with each other to launch attacks, we should impose a connectivity requirement so that the condition 
in Lemma~\ref{lemma_share} holds for any possible combination of pairs of neighboring malicious nodes in the network. 
The following proposition is the main result of this section. 

\begin{proposition}\label{theorem_share} 
	Consider the undirected network $\mathcal{G}=(\mathcal{V},\mathcal{E})$ under the $f$-total malicious model. Suppose that Assumptions \ref{neighborinfo}, \ref{cannotadd}, and \ref{broadcast} hold. Then, for Algorithm~1 with Detection Algorithm~2, the following hold:

	(a) All malicious nodes that behave against the averaging in Algorithm~1 are detected if and only if 
	for every pair of neighboring nodes, they have at least $f-1$ common neighbors.
	
	(b) Under the condition of (a), normal nodes achieve resilient average consensus if $\mathcal{G}$ is ($f+1$)-connected.
\end{proposition}

We proved in \cite{yuan2021secure} that under the $f$-total model, the graph condition in Lemma~\ref{lemma_share} is equivalent to condition (a) in Proposition~\ref{theorem_share}. Moreover, condition (b) in Proposition~\ref{theorem_share} guarantees that the normal network is connected. Thus, normal nodes using Algorithm~1 with Detection Algorithm~2 can achieve resilient average consensus as we proved in Proposition~\ref{thm_consensus}.

As we will further explain in Section~\ref{sec_sim_undirected}, the graph condition for Algorithm~2 does not require dense graph structures.
However, this feature is achieved at the cost of additional authentication from the secure mobile agents.

\section{Fully Distributed Detection in Directed Networks}\label{Secfors2}

In this section, we present our second distributed detection algorithm as Algorithm~3. It is fully distributed and operates without outside authentication for resilient average consensus in general directed networks. 
This important feature is realized through majority voting \cite{blahut1983theory, parhami1994voting} and requiring a denser graph structure. Moreover, we prove a necessary and sufficient graph condition for Algorithm~3 to properly function.

\begin{figure}[]
	\centering
	
	\hspace{9pt}
	\subfigure[]{
		\includegraphics[width=1.15cm]{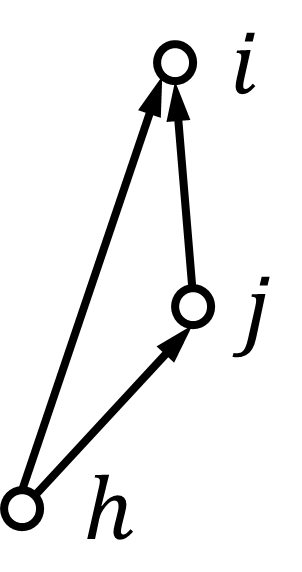}
	}
	\quad
	\hspace{5pt}
	\subfigure[]{
		\includegraphics[width=2.0cm]{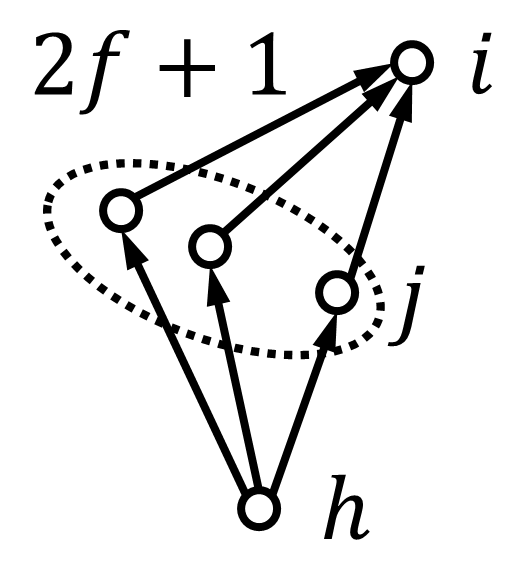}
	}
	
	\vspace{-3pt}
	\caption{Illustration of the graph condition for node $h$ being detectable by node $i$ in Definition~\ref{detectable}.}
	\label{fig_detectable}
\end{figure}

\subsection{Detection Algorithm Design}

In the last section, we have seen that node $i$'s information set consists of two parts that need to be investigated by its out-neighbors: (i) the current value $\overline{y}_i[k]$ ($\overline{z}_i[k]$) to check if it is updated according to Algorithm~1; (ii) the past values $\delta_{ij}[k]$ ($\omega_{ij}[k]$) used as inputs for updates to check if they are equal to the true values of the corresponding nodes.
In Algorithm~2, normal node $i$ can check whether part of the past values are manipulated in the information sets of its neighbors. More specifically, node $i$ knows the true past values of its direct neighbors. Then, using the sharing detection function, node $i$ can report a malicious node (or a pair of neighboring malicious nodes) if any values known by itself are manipulated.

However, to achieve fully distributed detection without any outside authentication, node $i$ should be able to independently verify whether any entries of the past values are manipulated in the information sets of its in-neighbors. 
We have seen that node $i$ can directly obtain the original value and the detection information of an in-neighbor $j$. For other values that node $i$ cannot directly obtain, we impose a certain graph structure so that it can access the original value and the detection information of a two-hop in-neighbor $h$ through majority voting. 
Specifically, if node $i$ receives $m$ values of node $h$, among the $m$ values, if more than $m/2$ values are the same, then node $i$ will take it as the true value of node $h$. 
In computer science, such redundancy schemes are common strategies to enhance the security and reliability of systems \cite{blahut1983theory,parhami1994voting}.

Next, we formally introduce the notion of \textit{detectable} nodes to indicate the kind of nodes that can be detected by node $i$ using Algorithm~3.

\begin{definition} \label{detectable}
	In the directed graph $\mathcal{G} = (\mathcal{V},\mathcal{E})$ under the $f$-local malicious model, node $h$ is said to be \textit{detectable} by node $i$ if one of the following conditions holds:
	\begin{itemize}
		\item $h\in \mathcal{N}_i^-$;
		\item there are at least $2f+1$ two-hop paths from $h$ to $i$.
	\end{itemize}
\end{definition}

We illustrate the above graph condition in Fig.~\ref{fig_detectable}.
Here, we also say that there is a detectable path from node $h$ to node $i$ if node $h$ is detectable by node $i$.
To achieve fully distributed detection, we need to impose a certain graph structure so that each node can have access to the necessary information used in its neighbors' updates (see Fig.~\ref{condition_s2}). We introduce the graph condition for Algorithm~3 as follows.

\begin{assumption} \label{condition2}
	A directed graph $\mathcal{G} = (\mathcal{V},\mathcal{E})$ under the $f$-local malicious model satisfies all the following conditions for $\forall i \in \mathcal{V}$:
	\begin{enumerate}
		\item any two-hop in-neighbor $h$ is detectable by $i$;
		\item any out-neighbor $q$ is detectable by $i$;
		\item any out-neighbor $l$ of in-neighbor $j$ is detectable by $i$.
	\end{enumerate}
	We will refer to conditions 1)--3) together as the graph condition for Algorithm 3.
\end{assumption}

\begin{figure}[t]
	\centering

	\hspace{45pt}
	\includegraphics[width = 6cm ]{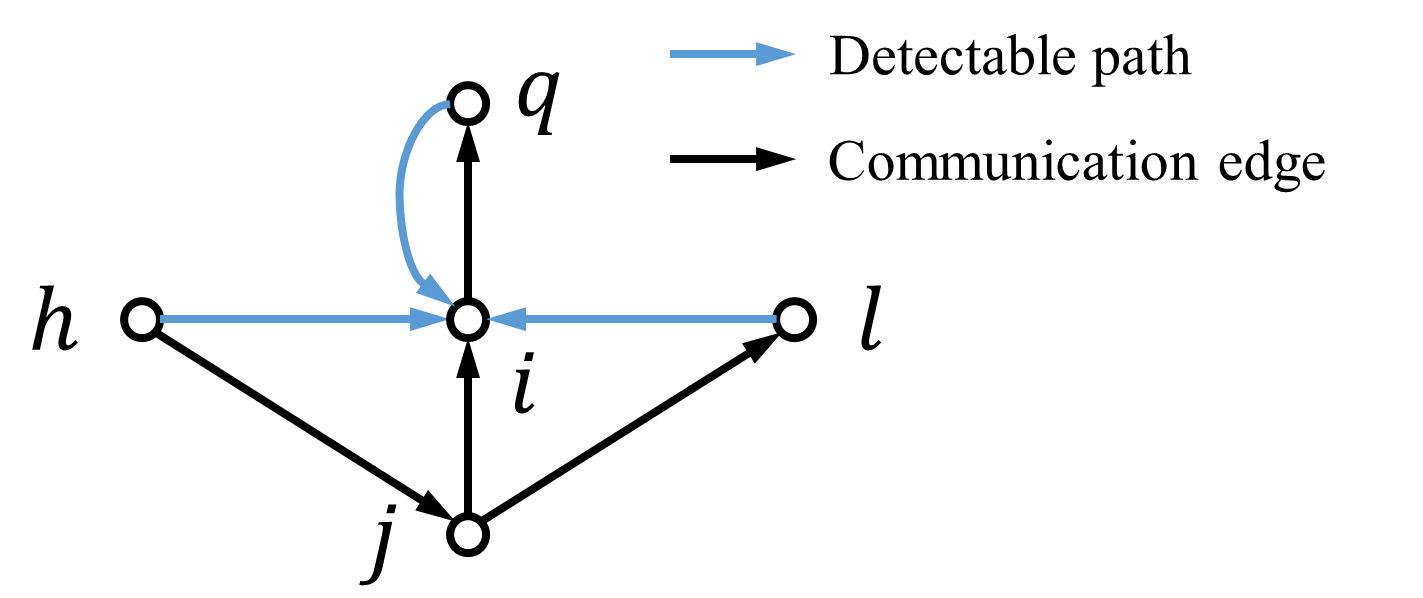}
	
	\vspace{-3pt}
	\caption{Illustration of the graph condition for Algorithm~3.} 
	\label{condition_s2}
\end{figure}

Although the above graph condition may require a locally dense graph structure, such a graph does not necessarily have a small diameter. This means that the path length of the shortest path between any two nodes may not be small. In fact, we can construct graphs satisfying the conditions in large scales. See the examples in Section~\ref{construction}.

We now present Algorithm~3. 
Each normal node $i$ performs majority voting on two things: the nodes' values and the detection information. Since we consider the $f$-local model in this section, if node $i$ receives the same information from at least $f+1$ distinct in-neighbors, it considers this information trustable. Additionally, node $i$ keeps a local set (only accessible to $i$) for the detection of two-hop in-neighbors as $\mathcal{A}_i^2[k]$ at time $k$.
After obtaining the true values of its one-hop in-neighbors $\forall j\in \mathcal{N}_i^-$ and two-hop in-neighbors $\forall h \in \mathcal{N}_i^{2-}$, it follows similar procedures as the ones in Algorithm~2.

\begin{algorithm}[t]
	\caption{Fully Distributed Detection Algorithm} 
	\LinesNumbered 
	\KwIn{$\Phi_j[k-1], \forall j\in \mathcal{N}_i^- \cup \{i\} $}
	
	\SetKwBlock{newbox}{Initialization:}{}
	\newbox{
		\SetAlgoVlined
		
		Node $i$ follows the initialization in Algorithm~2.
		
		Let $\mathcal{A}_i^2[1]= \emptyset$.
	}
	
	\For{ $k\geq 2$ }{
		
		\SetKwBlock{newbox}{Majority voting:}{}
		\newbox{
			\SetAlgoVlined
			
			Let $\mathcal{C}_i[k-1]=\mathcal{C}_i[k-1] \cup \{\delta_{jh}[k-1|k-1]$, $ \omega_{jh}[k-1|k-1], \forall h \in \mathcal{N}_i^{2-} \}$, i.e., the majority values of $h$ from $ \Phi_j[k-1]$, $j\in \mathcal{N}_i^-$.
			
			Let $\mathcal{A}_i[k]=\mathcal{A}_i[k-1]$, $\mathcal{A}_i^2[k]=\mathcal{A}_i^2[k-1]$. 
			
			\For{ ID $m\in $ majority of $\bigcup_{j\in \mathcal{N}_i^- } \mathcal{A}_j[k-1], $}{
				\eIf{   $m \in \mathcal{N}_i^-\cup \mathcal{N}_i^+$  }{
					let  $m \in \mathcal{A}_i[k]$, the faulty reporter $j' \in \mathcal{A}_i[k]$;
				}
				{let  $m \in \mathcal{A}_i^2[k]$, $j' \in \mathcal{A}_i[k]$.}
			}
		}
		
		\For{ $j\in \mathcal{M}_i^-[k]$ }{
			
			(Step 1a) \If{$\mathcal{N}_j^- \cap \mathcal{A}_j[k-1]$ contains any different detection of the nodes in $\mathcal{A}_i[k]\cup \mathcal{A}_i^2[k]$  }{
				let $j \in \mathcal{A}_i[k]$.
			}
			
			(Step 1b) \If{$\mathcal{N}_j^{2-} \cap \mathcal{A}_j[k-1]$ contains any ID not in $\mathcal{A}_i[k]\cup\mathcal{A}_i^2[k]$ or it does not contain a same ID in $\mathcal{A}_i[k]\cup\mathcal{A}_i^2[k]$ for the second time  }{
				let $j \in \mathcal{A}_i[k]$.
			}
			
			\textbf{Steps 2-4 in Algorithm 2} 
		}
		\KwOut{$\mathcal{A}_i[k]$}
		\textbf{Store:} $\delta_{jj}[k|k-1] $ and $\omega_{jj}[k|k-1] $ from $ \Phi_j[k-1]$, $ j\in \mathcal{N}_i^- \cup \{i\}$, into $\mathcal{C}_i[k]$.
	}
\end{algorithm}

\subsection{Necessary Graph Structure for Algorithm 3}

In Algorithm~3, we must impose the connectivity requirement in Assumption~\ref{condition2} on every node and its in-/out-neighbors. This enables the detection to be guaranteed for any combination of nodes being malicious in the network. The following theorem is the main result of this section.

\begin{theorem}\label{theorem_detect2} 
	Consider the directed network $\mathcal{G}=(\mathcal{V},\mathcal{E})$ under the $f$-local malicious model. Suppose that Assumptions \ref{neighborinfo} and \ref{cannotadd} hold. Then, for Algorithm~1 with Detection Algorithm~3, the following hold:
	
	(a) Each normal node detects all malicious nodes in its out-neighbors and in-neighbors within two hops that behave against the averaging in Algorithm~1 if and only if $\mathcal{G}$ satisfies the condition for Algorithm~3 (in Assumption~\ref{condition2}).
	
	(b) Under the condition of (a), normal nodes achieve resilient average consensus if $\mathcal{G}$ is $(f+1)$-strongly connected.
\end{theorem}

\begin{proof} \emph{(a) Necessity:} 
	We prove condition 1) by contradiction; conditions 2) and 3) follow by a similar proof.
	Suppose that there is a node $h\in \mathcal{N}_j^-$ with $h\notin \mathcal{N}_i^-$, and that there are at most $2f$ two-hop paths from node $h$ to node $i$ including the path containing node $j$. Suppose that node $j$ is malicious and there are $f-1$ malicious middle nodes in the paths from $h$ to $i$ (by the assumption of the $f$-local model).
	Then, in the worst case, node $i$ could get $f$ copies of the true value of $\delta_{hh}[k]$ from the $f$ normal middle nodes. In the mean time, node $i$ also gets $f$ copies of an identical false value of $\delta_{hh}[k]$ from the $f$ malicious middle nodes (including $j$). Thus, node $i$ cannot get majority regarding the true value of $\delta_{hh}[k]$. Thus, it cannot detect node $j$'s manipulation on $\delta_{jh}[k|k]$ in $\Phi_j[k]$.

	\emph{Sufficiency:} 
	By Definition~\ref{detectable}, if an out-neighbor $q$ of node $i$ is detectable by node $i$, then node $q$ becomes a direct in-neighbor or a two-hop in-neighbor of $i$. Therefore, the detection of node $q$ is the same as the detection of node $j$ or $h$ below.
	
	We must show that node $i$ can confirm the true value of every entry of the information set $\Phi_j[k]$ of in-neighbor $j$ by three major steps at time $k+1$. See the illustration in Fig.~\ref{condition_s2}. First, it can obtain the true values $\delta_{hh}[k]$, $ \omega_{hh}[k] $ of every neighbor $h\in \mathcal{N}_j^-$ (i.e., $i$'s two-hop in-neighbor $h$) from the previous time $k$. Second, node $i$ can obtain the correct detection of $h\in \mathcal{N}_j^-$ before the detection loop at time $k$. Moreover, node $i$ can obtain the correct detection of $j$'s out-neighbor $l$ depending on the corresponding case ($l$ can be a two-hop in-neighbor or a direct in-neighbor of $i$ by condition 3) in Assumption~\ref{condition2}). Then we can prove that node $i$ will detect node $j$ at time $k+1$ if node $j$ sends out faulty $\Phi_j[k]$.

	Depending on how the detectable path is formed, consider two cases for node $i$: (i) nodes $h$ and $l$ are direct in-neighbors of $i$; (ii) nodes $h$ and $l$ are two-hop in-neighbors of $i$.
	
	(i) In the case where $h, l\in \mathcal{N}_i^-$, it is clear that node $i$ can receive the true values $\delta_{hh}[k]$ and $ \omega_{hh}[k]$ from $\Phi_h[k-1]$. Moreover, it can have the correct detection of its direct in-neighbors $h$ and $l$ before time $k$. 
	
	(ii) Suppose that $h\notin \mathcal{N}_i^-$, and there are at least $2f+1$ two-hop paths from $h$ to $i$. In this case, there is some normal node $p\in \mathcal{N}_h^+ \cap \mathcal{N}_i^-$ which carries the true values $\delta_{hh}[k]$ and $ \omega_{hh}[k]$ in its information set $\Phi_p[k]$. 
	Then, node $i$ can get the true values $\delta_{hh}[k]$ and $ \omega_{hh}[k]$ since the majority of the $2f+1$ paths from $h$ to $i$ contain nodes as $p$ by the $f$-local model.

	For node $i$ to obtain the correct detection of two-hop in-neighbors $h$ and $l$, it follows a similar analysis. We look at the case for $h$.
	If $h$ transmits faulty $\Phi_h[k-1]$, then it is detected by its one-hop neighbors at time $k$. Recall that there are at least $2f+1$ directed two-hop paths from $h$ to $i$. Thus, under the $f$-local model, node $i$ can obtain the correct detection of $h$ by majority voting before the detection loop of time $k+1$.

	Therefore, node $i$ knows the true values $\delta_{hh}[k]$, $ \omega_{hh}[k]$ and obtains the correct detection of its two-hop in-neighbors $h$, $l$ before running the detection loop at time $k+1$. Thus if node $j\in \mathcal{N}_i^-$ sends out faulty $\Phi_j[k]$ by possible manipulation including modifying $\delta_{jh}[k]$ and/or $ \omega_{jh}[k]$ in $\Phi_j[k]$, or by sending false detection information of nodes $h$ and $l$, then node $i$ will detect. Note that when out-neighbor $l$ is a two-hop in-neighbor of $j$, the detection of $l$ is included in $\mathcal{A}_j[k+1]$ and is sent to node $i$ in $\Phi_j[k+1]$. Therefore, step 1b in Algorithm~3 is designed to handle this case. This procedure will not cause problems since the removal of malicious neighbors can be asynchronous at each normal agent.

	Next, we show that node $i$ can verify if node $j$ updates $\delta_{jj}[k+1]$ and  $\omega_{jj}[k+1]$ in $\Phi_j[k]$ by the averaging in Algorithm~1 or not. This is done by reconstructing $\lambda_j'[k+1]$ and $\gamma_j'[k+1]$ and checking whether $\epsilon_\lambda=\epsilon_\gamma=0$, where
	\vspace*{-1.0mm}
	\begin{equation*}
		\begin{aligned} 
			\delta_{jj}[k+1]&=\lambda_j'[k+1]+\epsilon_\lambda,\\
			\omega_{jj}[k+1]&=\gamma_j'[k+1]+\epsilon_\gamma.
		\end{aligned}
	\end{equation*}
	As shown above, node $i$ can verify $\mathcal{A}_j[k]$ in $\Phi_j[k]$. Thus,
	\begin{equation*}
		\begin{aligned} 
			\lambda_j'[k] \medspace\medspace &= \medspace\medspace \delta_{jj}[k-1|k-1] \\ 
			& \medspace\medspace\medspace +  \frac{y_j'[k-1] + |\Delta \mathcal{M}_j^+[k-1]|\delta_{jj}[k-1|k-1]}{1+d_j^+[k-1]} ,  \\
			\gamma_j'[k] \medspace\medspace &= \medspace\medspace \omega_{jj}[k-1|k-1] \\ 
			& \medspace\medspace\medspace +  \frac{z_j'[k-1] + |\Delta \mathcal{M}_j^+[k-1]|\omega_{jj}[k-1|k-1]}{1+d_j^+[k-1]} ,  \\
		\end{aligned}
	\end{equation*}
	where $\delta_{jj}[k-1|k-1]$, $\omega_{jj}[k-1|k-1]$, $\Delta \mathcal{M}_j^+[k-1]$, $d_j^+[k-1]$ are from $\Phi_j[k-1]$ with $\mathcal{A}_j[k-1]$. Moreover, $y_j'[k-1]$ and $z_j'[k-1]$ are obtained by node $i$ through
	\begin{equation*}
		\begin{aligned} 
			y_j'[k-1] \medspace\medspace &= \medspace\medspace (\delta_{jj}[k|k-1]-\delta_{jj}[k-1|k-1]) \\ 
			& \medspace\medspace\medspace \times (1+ d_j^+[k-1]), \\
			z_j'[k-1] \medspace\medspace &= \medspace\medspace (\omega_{jj}[k|k-1]-\omega_{jj}[k-1|k-1]) \\ 
			& \medspace\medspace\medspace \times (1+ d_j^+[k-1]). \\
		\end{aligned}
	\end{equation*}	
	
	We also note that node $i$ has access to the true values $\delta_{hh}[k]$ and $ \omega_{hh}[k]$. Besides, node $i$ knows $\delta_{jh}[k-1]=\delta_{hh}[k-1]$ and $ \omega_{jh}[k-1]=\omega_{hh}[k-1]$ from $\Phi_j[k-1]$. Otherwise, node $j$ would have been detected at time $k$. Thus, we have 
	\begin{equation*}
		\begin{aligned}
			\lambda_j'[k+1] \medspace\medspace &=\medspace\medspace \lambda_j'[k] + \sum\limits_{h\in \mathcal{N}_j^- \cup \{j\}}{(\delta_{jh}[k]-\delta_{jh}[k-1])},  \\
			\gamma_j'[k+1] \medspace\medspace &=\medspace\medspace \gamma_j'[k] + \sum\limits_{h\in \mathcal{N}_j^- \cup \{j\}}{(\omega_{jh}[k]-\omega_{jh}[k-1])},  
		\end{aligned}
	\end{equation*}
	where $\delta_{jh}[k]=0$ and $\omega_{jh}[k]=0$ for $h\in \mathcal{A}_j[k]$. Then node $i$ can compare $\delta_{jj}[k+1]$ (and $\omega_{jj}[k+1]$) in $\Phi_j[k]$ with $\lambda_j'[k+1]$ (and $\gamma_j'[k+1]$) and checks if node $j$ follows the averaging in Algorithm~1 or not.

	\emph{(b)} A malicious node will be detected immediately once it manipulates its information set. Thus, misbehavior of any malicious node cannot affect normal nodes since normal nodes exclude values from detected malicious nodes in Algorithm~1. Moreover, since $\mathcal{G}$ is $(f+1)$-strongly connected, the normal network is strongly connected after removing the $f$-local adversary node set. Therefore, resilient average consensus is achieved as shown in Proposition~\ref{thm_consensus}.
\end{proof}

Here, we show that the graph satisfying the condition for Algorithm~3 has the minimum in-degree as $2f+1$. It indicates that we need to make the minimum in-degree no less than $2f+1$ when we design a desirable network topology.
Conversely, it is straightforward that a graph with minimum in-degree less than $2f+1$ does not meet our condition.
We formally state the property in the next lemma for general strongly connected digraphs since strong connectivity is necessary for achieving average consensus in directed graphs \cite{hadjicostis2015robust}.
Moreover, we can confirm that a complete graph $\mathcal{K}_n$ satisfies the condition for Algorithm~3. To avoid trivial cases, we consider $n>3$ in the following discussions.

\begin{lemma} \label{directed_minimum_degree}
	If a strongly connected and incomplete digraph $\mathcal{G} = (\mathcal{V},\mathcal{E})$ (under the $f$-local model) satisfies the condition for Algorithm~3 in Assumption~\ref{condition2}, then $\mathcal{G}$ has the minimum in-degree no less than $2f+1$.
\end{lemma}

The proof of Lemma~\ref{directed_minimum_degree} can be found in the Appendix.

We note that the detection condition (a) for Algorithm~3 is not sufficient to guarantee strong connectivity of the graph. A simple counter example is a graph with two disconnected complete subgraphs. Clearly, the whole graph is not connected while it meets the detection condition for Algorithm~3. This observation reveals that the consensus condition (b) guaranteeing the normal network to be strongly connected is also critical for our RAC algorithm.

\subsection{Graph Condition in Undirected Networks}\label{discuss2}

For the special case of undirected networks, the condition for Algorithm 3 can be simplified as follows.

\begin{lemma} \label{undirected_condition2}
	An undirected graph $\mathcal{G} = (\mathcal{V},\mathcal{E})$ satisfies the condition for Algorithm~3 in Assumption~\ref{condition2} if for each node $i \in \mathcal{V}$, it holds that any two-hop in-neighbor $h$ of node $i$ is detectable by node $i$.
\end{lemma}

\begin{proof}
	We can easily observe that condition 2) in Assumption~\ref{condition2} is satisfied automatically in undirected networks. Moreover, it holds that an out-neighbor $l$ of node $i$'s in-neighbor $j$ is a two-hop in-neighbor of node $i$ in undirected networks. Therefore, condition 3) can be derived if condition 1) holds in an undirected network.
\end{proof}

Next, we show that for undirected networks, the condition for Algorithm~3 and connectivity of the graph together guarantee that the normal network is connected after removing the $f$-local malicious node set.

\begin{proposition} \label{undirected_simple}
	Consider the undirected graph $\mathcal{G} = (\mathcal{V},\mathcal{E})$ under the $f$-local model. If (i) $\mathcal{G}$ is connected and (ii) for any node $i \in \mathcal{V}$, it holds that any two-hop in-neighbor $h$ of node $i$ is detectable by node $i$, then the normal network induced by the normal agents in $\mathcal{N} \subseteq \mathcal{V}$ is connected.
\end{proposition}

The proof of Proposition~\ref{undirected_simple} can be found in the Appendix.

We can see from the above two results that both the detection condition (a) and consensus condition (b) for undirected networks are simplified compared to the conditions in Theorem~\ref{theorem_detect2} for directed networks.
We formally state the conditions for undirected networks as follows, which can be proved by Lemma~\ref{undirected_condition2} and Proposition~\ref{undirected_simple}.

\begin{theorem}\label{detect2_undirected} 
	Consider the undirected network $\mathcal{G}=(\mathcal{V},\mathcal{E})$ under the $f$-local malicious model. Suppose that Assumptions \ref{neighborinfo} and \ref{cannotadd} hold. Then, for Algorithm~1 with Detection Algorithm~3, the following hold:

	(a) Each normal node detects all malicious neighbors within two hops that behave against the averaging in Algorithm~1 if and only if for any node $i \in \mathcal{V}$, it holds that any two-hop neighbor $h$ of node $i$ is detectable by node $i$.
	
	(b) Under the condition of (a), normal nodes achieve resilient average consensus if $\mathcal{G}$ is connected.
\end{theorem}

\subsection{Construction of Graphs Satisfying the Condition}\label{construction}

\begin{figure}[t]
	\centering
	
	\subfigure[\scriptsize{The four-layer directed graph which contains an undirected graph after removing the red directed edges.}]{
		\includegraphics[width = 4.6cm ]{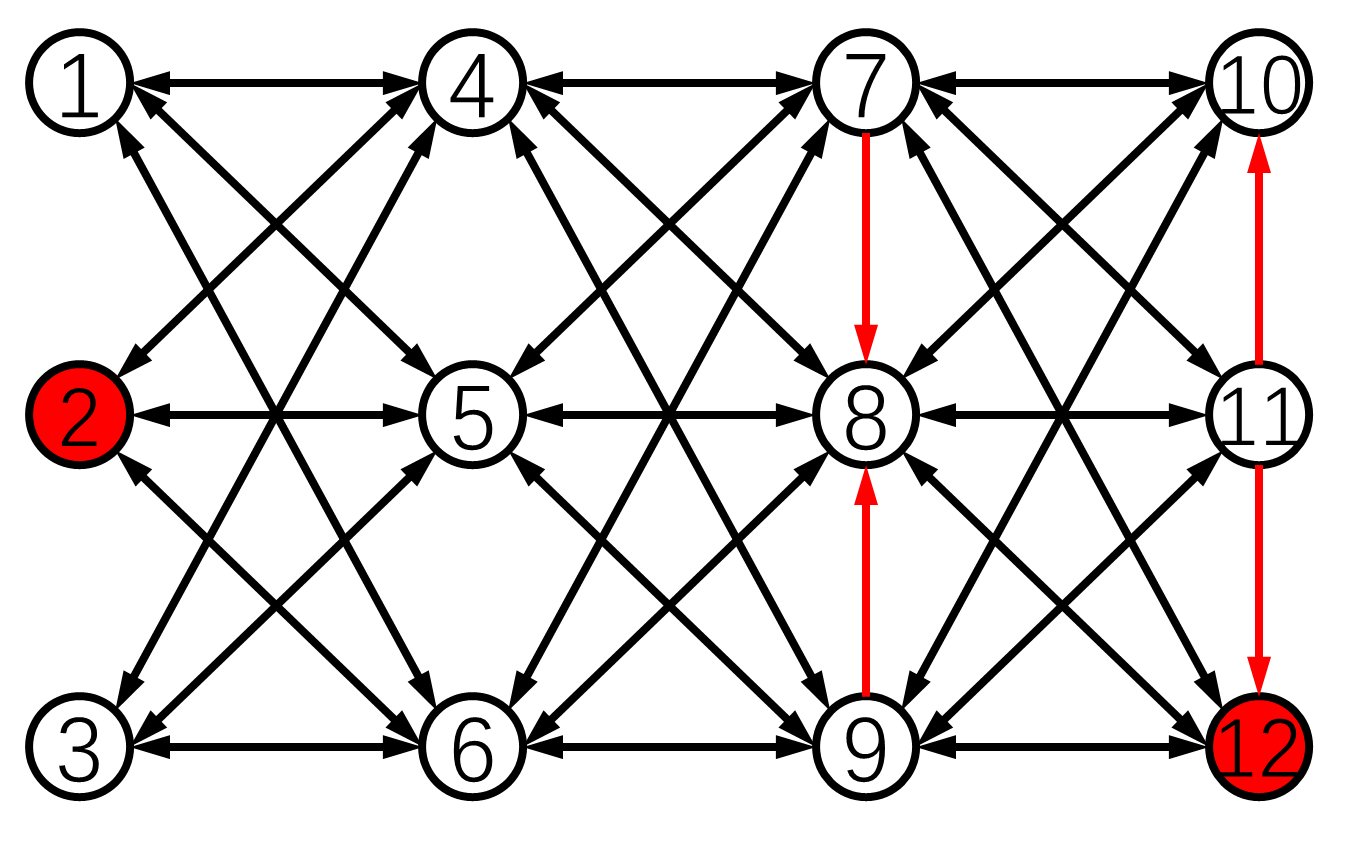}
	}
	
	\vspace{-3pt}
	
	\subfigure[\scriptsize{The five-layer directed graph.}]{
		\includegraphics[width = 5.7cm ]{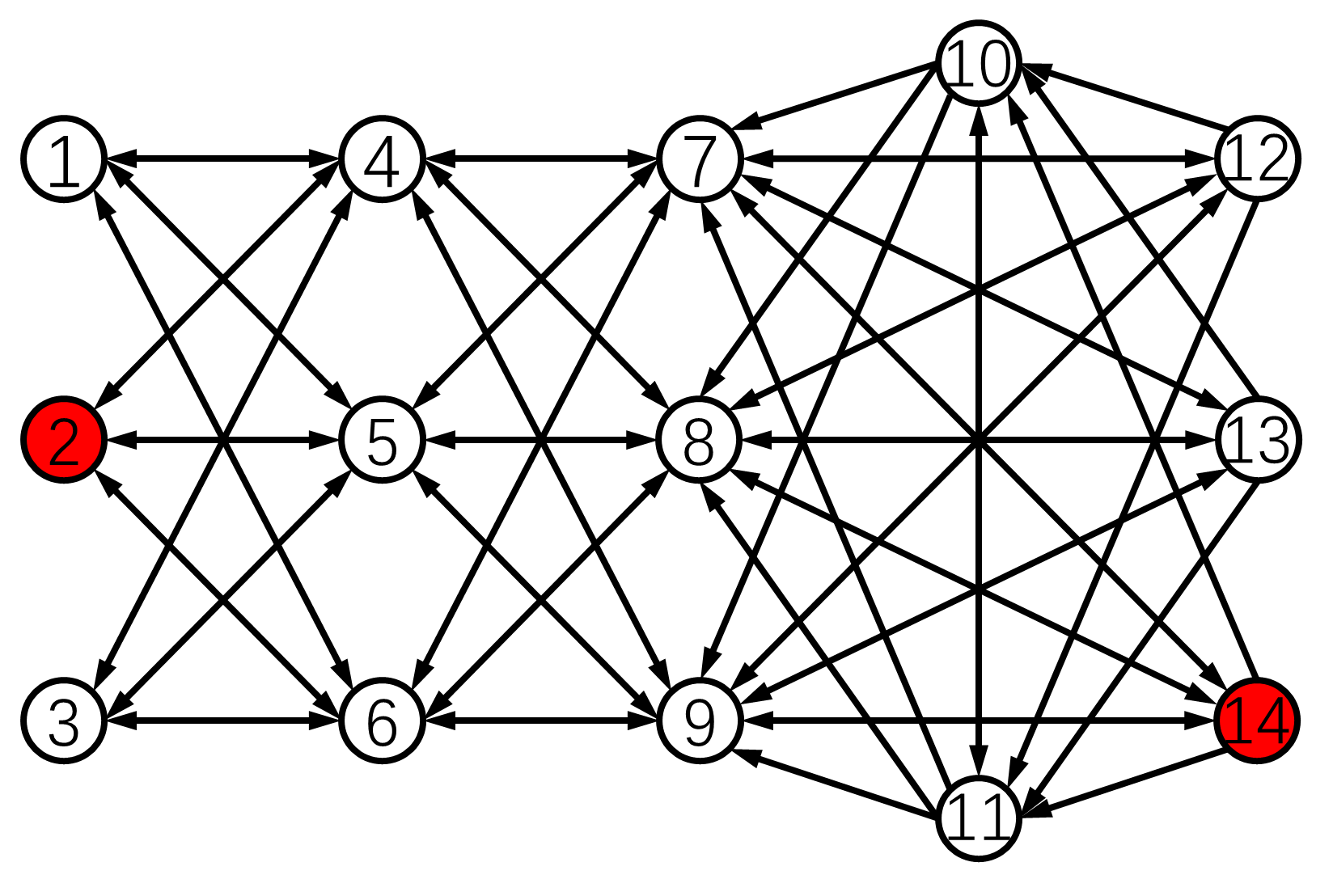}
	}
	
	\vspace{-3pt}
	\caption{Large-scale graphs satisfying the condition for Algorithm~3 under the $f$-local model. Here, we set $f=1$ for illustration.}
	\label{large_graph}
\end{figure}

In this part, we present some example graphs satisfying the conditions for our algorithms. Furthermore, we present a systematic way to construct large-scale graphs that meet the condition for Algorithm~3.

We first give example graphs satisfying the conditions in Theorem~\ref{theorem_detect2}.
The network in Fig.~\ref{large_graph}(a) satisfies the conditions for Algorithm~3 under the $1$-local model.
Moreover, there is a characteristic four-layer structure. We can extend this idea to the cases with any $f$. Each layer should have $2f+1$ nodes for the $f$-local model. Each node in one layer should be connected with every node in the neighbor layers and have no connection with the nodes in its own layer. This structure can also have many layers as long as the $f$-local set is met for $\forall i\in \mathcal{N}$. A similar way for constructing large-scale directed networks for Algorithm~3 is presented in Fig.~\ref{large_graph}(b).
From these examples, we can conclude that the unbalanced directed graphs can also meet the condition for Algorithm~3.
Simulations of Algorithm~3 in these networks will be given in Section~\ref{sec_simulation}.

Node $i$ is said to be a full access node if it is an out-neighbor of all other nodes in the network \cite{yuan2021secure}. Notice that such a node can detect any malicious node in the network.
We can enhance the performance of Algorithm~3 (and Algorithm~2) by incorporating nodes with such characteristics. 
However, we emphasize that we do not assume such full access nodes to be normal. As long as the conditions for Algorithm~3 (or Algorithm~2) are met, a full access node can also be detected by its normal neighbors when it misbehaves. This result can be easily proved by Theorem \ref{theorem_detect2} and is given as follows.

\begin{corollary}\label{fullaccessnode}
	A normal full access node using Algorithm~3 (or Algorithm~2) detects any node behaving against the averaging in Algorithm~1 in the network.
\end{corollary}

As an example, the 5-node undirected network in Fig.~\ref{small_graph}(a) could tolerate two malicious nodes when the conditions for $1$-local are met except for the full access node 1. In the same graph, if only node 1 becomes malicious and the conditions for $1$-local are also met for other nodes, then RAC is still guaranteed. As another example, in the 8-node directed network in Fig.~\ref{8node}, normal nodes using Algorithm~1 with Detection Algorithm~3 can achieve resilient average consensus even in the presence of 5 malicious nodes. More details are discussed in the numerical examples.

The following corollary states the maximum tolerable number of malicious nodes in a network applying our algorithms. It can be directly proved from Theorem~\ref{theorem_detect2}.

\begin{corollary}\label{completegraph2}
	In the complete graph $\mathcal{K}_n$, normal nodes using Algorithm~1 with Algorithm~3 (or Algorithm~2) achieve resilient average consensus in the presence of $f\leq n-2$ malicious nodes in the network.
\end{corollary}

\begin{figure}[t]
	\centering
	
	\subfigure[]{
		\includegraphics[width=2.0cm]{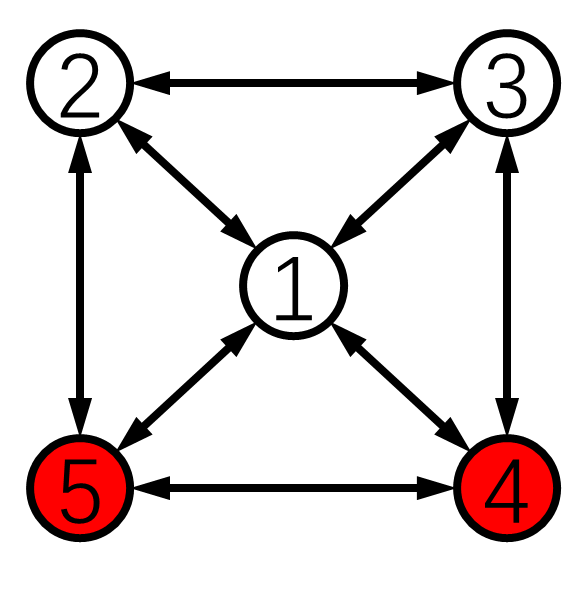}
	}
	\quad
	\subfigure[]{
		\includegraphics[width=3.25cm]{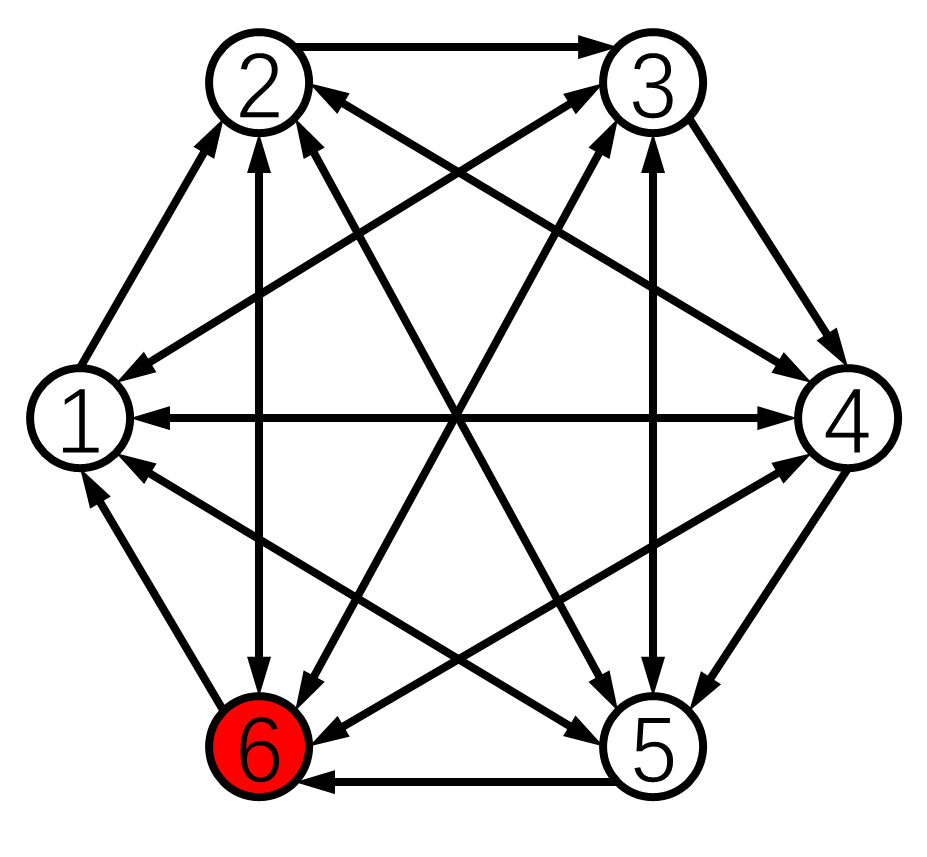}
	}
	
	\vspace{-3pt}
	\caption{Small-scale graphs satisfying the condition for Algorithm~3 under the $1$-local model (except for the full access node 1 in (a)).}\label{small_graph}
\end{figure}

\begin{remark}
	From above, we see that it is relatively easy to check whether a graph meets the condition in Assumption~\ref{condition2} since there is no combinatorial process in the checking. In particular, the verification of the detectable condition in Definition~\ref{detectable} is very simple. Further, the checking on the condition in Assumption~\ref{condition2} requires less than $n(n-1)$ times of such verification (usually much less for sparse graphs). In contrast, the verification of robustness of a graph in resilient consensus works (e.g., \cite{leblanc2013resilient,dibaji2018resilient,ishii2022overview}) involves combinatorial processes and is computationally NP-hard.
	Moreover, we have proved a much simpler condition for undirected networks using our algorithm, which is suitable for the deployment in large-scale networks. Besides, the proposed systematic way for constructing the desirable large-scale graphs can also facilitate the deployment of our algorithm in various applications.
\end{remark}

\subsection{Discussions and Comparisons with Related Works}\label{discuss22}

We discuss the differences between the resilient consensus (RC) algorithm from \cite{yuan2021secure} and RAC Algorithm~3 for directed networks.
In \cite{yuan2021secure}, normal nodes achieve resilient consensus by monitoring their in-neighbors. However, the situation becomes much more complex for Algorithm~3.
Aside from the detection of in-neighbors, node $i$ should also be able to detect each out-neighbor. This is because for solving the RAC problem, node $i$ should not send its $y$ and $z$ values to the malicious out-neighbor(s) so that the normal network can accurately preserve the ``mass'' of normal nodes only and achieve averaging as shown in Theorem~\ref{theorem_detect2}. The detection condition for Algorithm~3 hence requires denser graphs than the one for the RC algorithm \cite{yuan2021secure}.
Moreover, notice that the necessary condition for resilient average consensus is the strong connectivity of the graph. It requires each node to have at least one out-going edge. In contrast, the necessary condition for resilient consensus is that there is at least one rooted spanning tree in the graph. Therefore, the necessary condition for Algorithm~3 is stricter than the one for the RC algorithm \cite{yuan2021secure}.

The recent work \cite{dibaji2019resilient} proposed a certified propagation algorithm (CPA)-type broadcast and retrieval approach for the RAC problem.
There, each normal node broadcasts its initial value to all the nodes in the network through relaying by neighbors (i.e., the flooding technique). Then, the normal node confirms another node's initial value if it receives more than $f$ copies of the value of the same node, which is similar to the CPA approach \cite{tseng2015broadcast}. Lastly, the normal node converges to the average of values from the nodes which it has confirmed. 

We must note that this kind of approaches for each node to verify and store the initial values of all the normal agents become infeasible in large-scale networks, as it consumes intensive storage and computation for each single node to monitor the whole network. 
Compared to these approaches, our iterative detection algorithm is more efficient, especially in large-scale networks. Specifically, the storage needed on each node for our algorithm is modest as each node stores only local information of its in-neighbors within two hops.

It is challenging for our algorithms, as well as any other algorithms \cite{fagiolini2009dynamic,pasqualetti2012consensus, zhao2018resilient,hadjicostis2023identification} to identify adversarial nodes that adopt extreme initial values but behave according to the proposed algorithms as if they were normal nodes. Clearly, such nodes are indistinguishable from normal nodes with extreme initial values.
To mitigate the impact of such adversary nodes, we can set a safety interval $[y_{\min}, y_{\max}]$ (recall that $y_i[0]=x_i[0]$) for normal nodes so that neighbors taking initial values outside this interval are considered malicious \cite{yuan2021secure}.

We conduct some comparisons between Algorithms~2 and 3 for the case of undirected networks. Recall that Algorithm~2 is for the $f$-total model while Algorithm~3 is for the $f$-local model.
We first note that the $f$-local model contains the $f$-total model and is more adversarial in the sense that more than $f$ malicious agents in total may be in the entire network under the $f$-local model.
The reason is that if the graph condition for Algorithm~2 is satisfied under the $f$-local model, then it cannot guarantee that there is a normal neighbor of any pair of adjacent malicious nodes (Lemma~\ref{lemma_share}). Here is a simple counter example. Consider the 5-node network in Fig.~\ref{small_graph}(a) with malicious node set $\mathcal{A}=\{1,4,5\}$. It satisfies the $f-1$ common neighbors condition for Algorithm~2 under the $2$-local model. Yet, it does not meet the condition in Lemma~\ref{lemma_share}. 
However, we must note that this phenomenon is not in presence for Algorithm~3 since the condition in Theorem~\ref{detect2_undirected} has required the necessary condition for each node to independently detect the malicious neighbors.

\section{Numerical Examples}\label{sec_simulation}

We present numerical examples to verify the efficacy of RAC Algorithm~1 with Detection Algorithms~2 and~3.

\subsection{Simulations with Directed Networks}

In this part, we provide the simulations of Algorithm~3 in three directed networks of different scales.

\textit{1) Small Directed Network:} Consider the 6-node network in Fig.~\ref{small_graph}(b). It meets the graph condition (Assumption~\ref{condition2}) for Algorithm~3 under the 1-local model. Moreover, it is 2-strongly connected (i.e., the remaining graph is strongly connected after the removal of the 1-local malicious node set). Hence, it meets the requirements in Theorem~\ref{theorem_detect2}.

\begin{figure}[t]
\centering

\subfigure[\scriptsize{Under attacks with less edges.}]{
	\includegraphics[width=3.4in,height=1.4in]{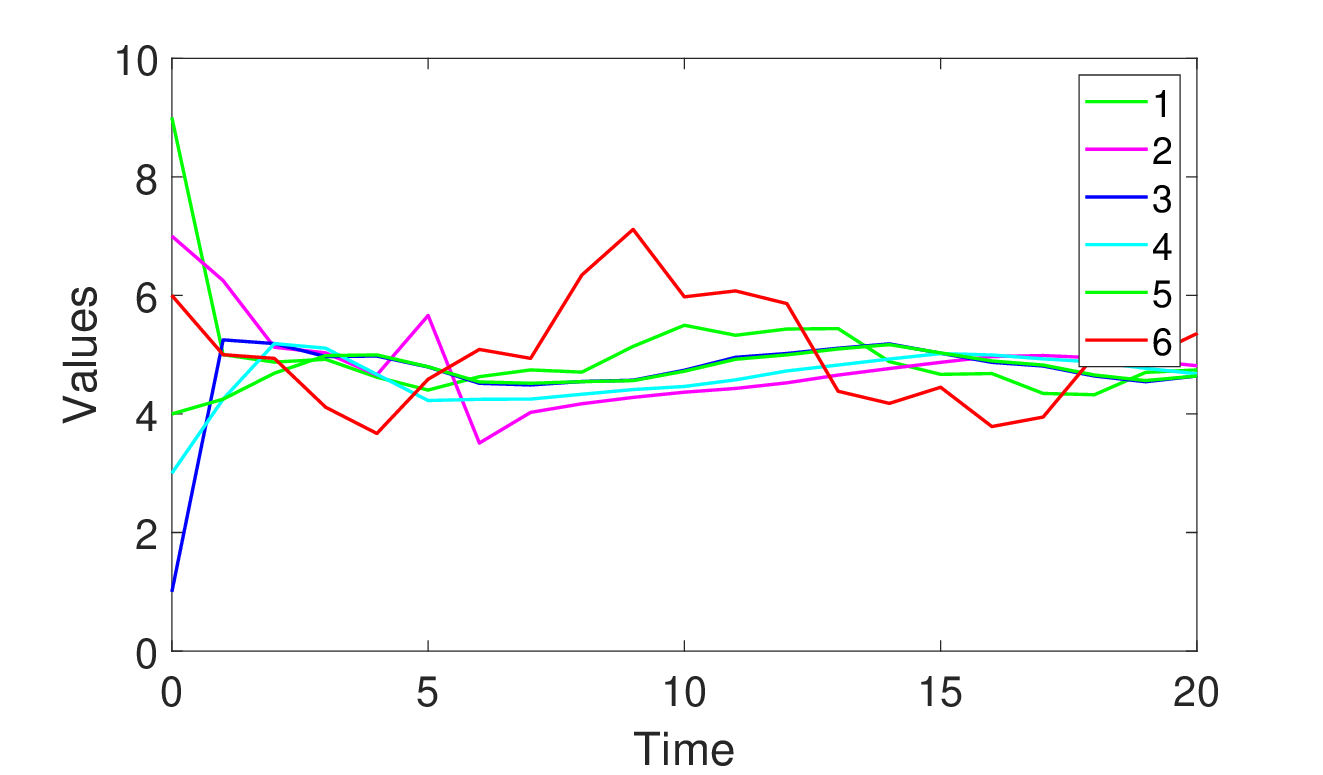}
}

\vspace{-3pt}

\subfigure[\scriptsize{Under attacks.}]{
	\includegraphics[width=3.4in,height=1.4in]{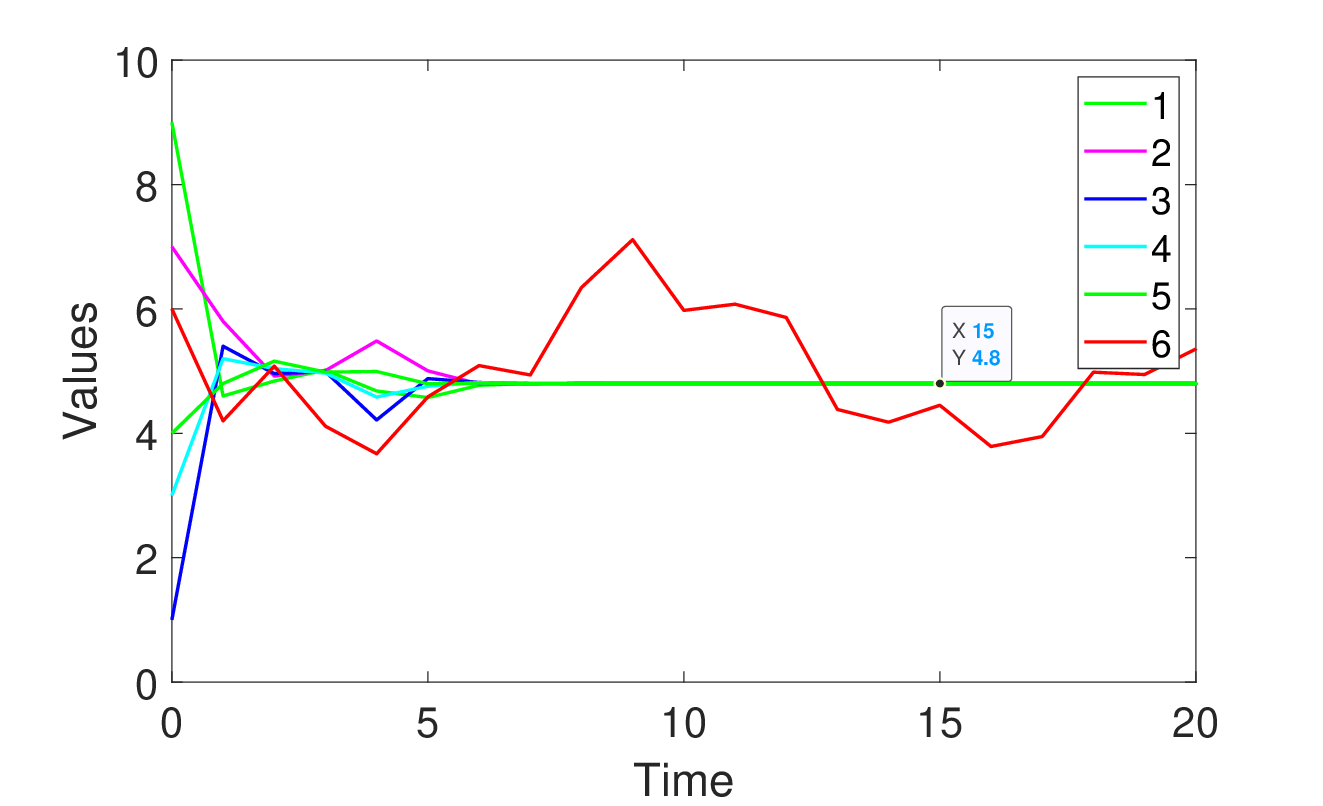}
}

\vspace{-3pt}
\caption{Time responses of Algorithm~3 in the 6-node network in Fig.~\ref{small_graph}(b).}
\label{A3-6node}
\end{figure}

Set the initial values $x[0]=[9\ 7\ 1\ 3\ 4\ 6]^T$ and the adversary node set $\mathcal{A}=\{6\}$. Malicious node 6 is indicated in red in Fig.~\ref{small_graph}(b).
First, we show that the detection condition for Algorithm~3 is critical for the success of our RAC algorithm. Suppose that three undirected edges $(1,4)$, $(2,5)$, and $(3,6)$ are removed from the network, i.e., the condition for Algorithm~3 is not satisfied. The simulation under attacks for the above case with less edges is displayed in Fig.~\ref{A3-6node}(a). Until time $k=2$, malicious node 6 follows the averaging in Algorithm~1 to avoid being detected. Then it manipulates its $y$ values through changing the past values of node 2 while not manipulating other entries of its information set. We can see in Fig.~\ref{A3-6node}(a) that resilient average consensus is not achieved by Algorithm~3 with less edges. This is because only nodes 2 and 4 can detect the above misbehavior of node 6. Normal nodes 1, 3 and 5 are misled by node 6 due to the lack of necessary graph structure to obtain the correct value of node 2.

Next, we apply Algorithm~3 in the 6-node network as presented in Fig.~\ref{small_graph}(b), where the condition for Algorithm~3 is met. Consider the same initial states and the same attacks for the network.
The simulation result is presented 
in Fig.~\ref{A3-6node}(b). Malicious node 6 launches attacks as before, however, this misbehavior is soon detected by its normal out-neighbors. The normal nodes then compensate the erroneous effects received from node 6 and start to form consensus among normal neighbors only.
Lastly, the normal nodes are able to reach the average of their initial values $\overline{X}_{\mathcal{N}}= \frac{\sum_{i\in\mathcal{N}} x_i[0]} {|\mathcal{N}|} = 4.8$, and resilient average consensus is reached using Algorithm~3.

\begin{figure}[t]
\centering

\subfigure[\scriptsize{No attack.}]{
	\includegraphics[width=3.4in,height=1.4in]{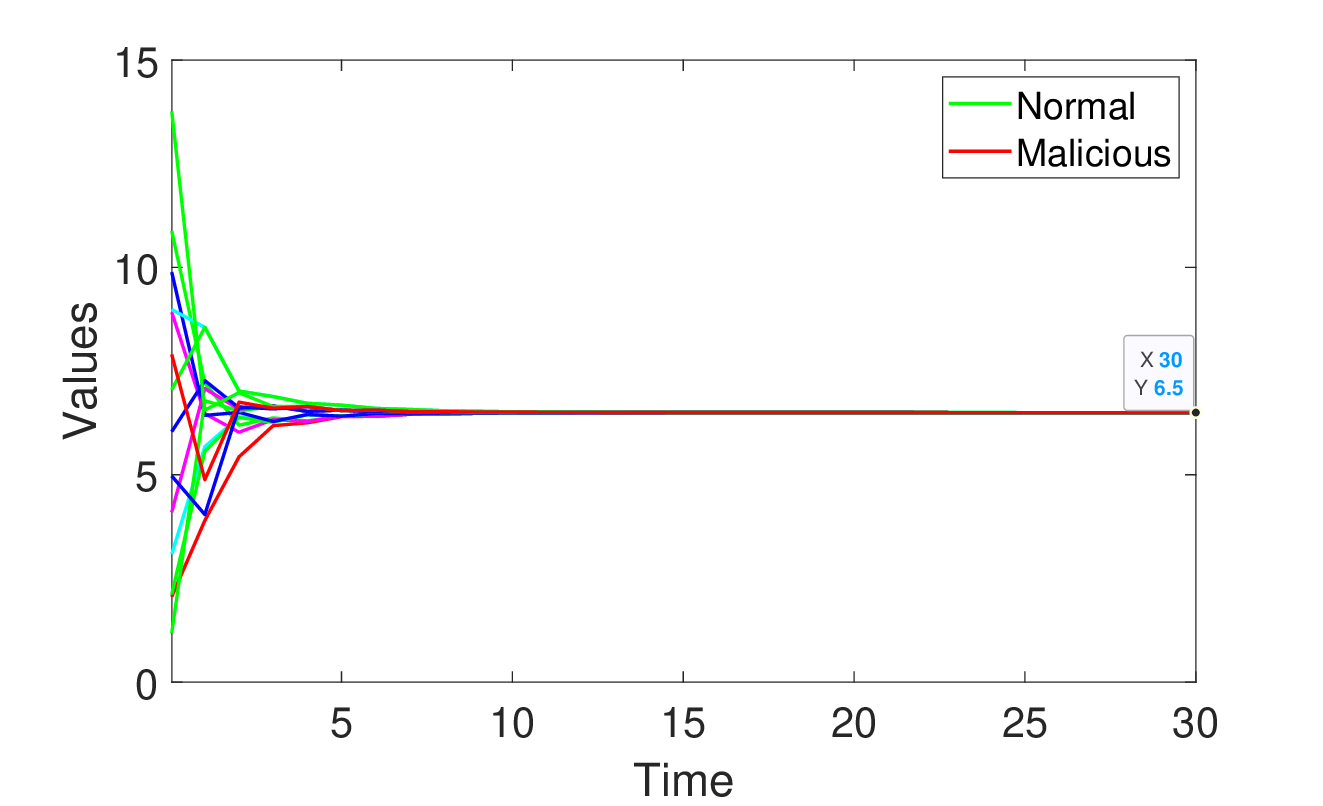}
}

\vspace{-3pt}

\subfigure[\scriptsize{Under attacks.}]{
	\includegraphics[width=3.4in,height=1.4in]{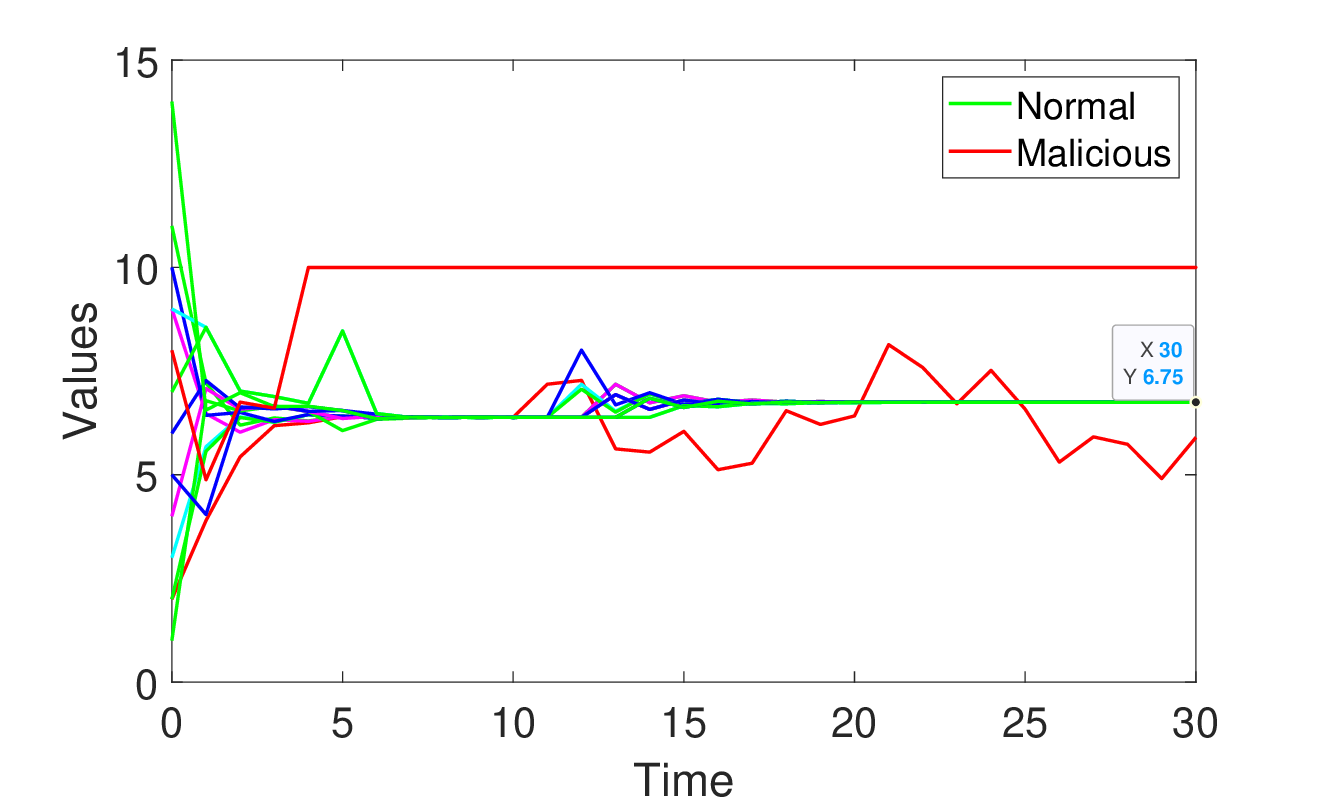}
}

\vspace{-3pt}

\caption{Time responses of Algorithm~3 in the 14-node network in Fig.~\ref{large_graph}(b).}
\label{A3-14node}
\end{figure}

\textit{2) Medium-scale Directed Network:} Next, consider the 14-node network in Fig.~\ref{large_graph}(b) constructed using the method in Section~\ref{construction}. It satisfies the condition for Algorithm~3 under the 1-local model and is 2-strongly connected. 

Let the initial values $x[0]=[11\ 2\ 9\ 3\ 2\ 10\ 1\ 4\ 6\ 9\ 7\ 5\ 14\ 8]^T$ and $\mathcal{A}=\{2, 14\}$ in Fig.~\ref{large_graph}(b). The time responses of nodes under no attack are shown in Fig.~\ref{A3-14node}(a), where all nodes using Algorithm~3 reach the average of their initial values $\overline{X}= \frac{\sum x[0]} {n} = 6.5$.
Here, the lines not in red represent the values of normal nodes.
Then, the time responses of nodes under attacks are displayed in Fig.~\ref{A3-14node}(b). There, until time $k=3$, malicious nodes~2 and~14 pretend to be normal by following the averaging. Then node 14 changes its own value to a fixed value and is detected by its normal out-neighbors at the next time step. In the meantime, node~2 keeps concealing itself. At time $k=11$, node 2 and normal nodes almost reach the average of their initial values (around 6.385). However, node 2 starts to manipulate its $y$ value through changing the past values of its in-neighbors in its information set. Such an attack is also quickly detected and normal nodes remove the effects received from node 2 until then. Finally, the normal nodes reach the average of their initial values $\overline{X}_{\mathcal{N}}=  6.75$, and resilient average consensus is attained.
Moreover, the convergence of Algorithm~3 is quick even in the presence of malicious attacks.

\begin{figure}[t]
\centering

\includegraphics[width = 2.9cm ]{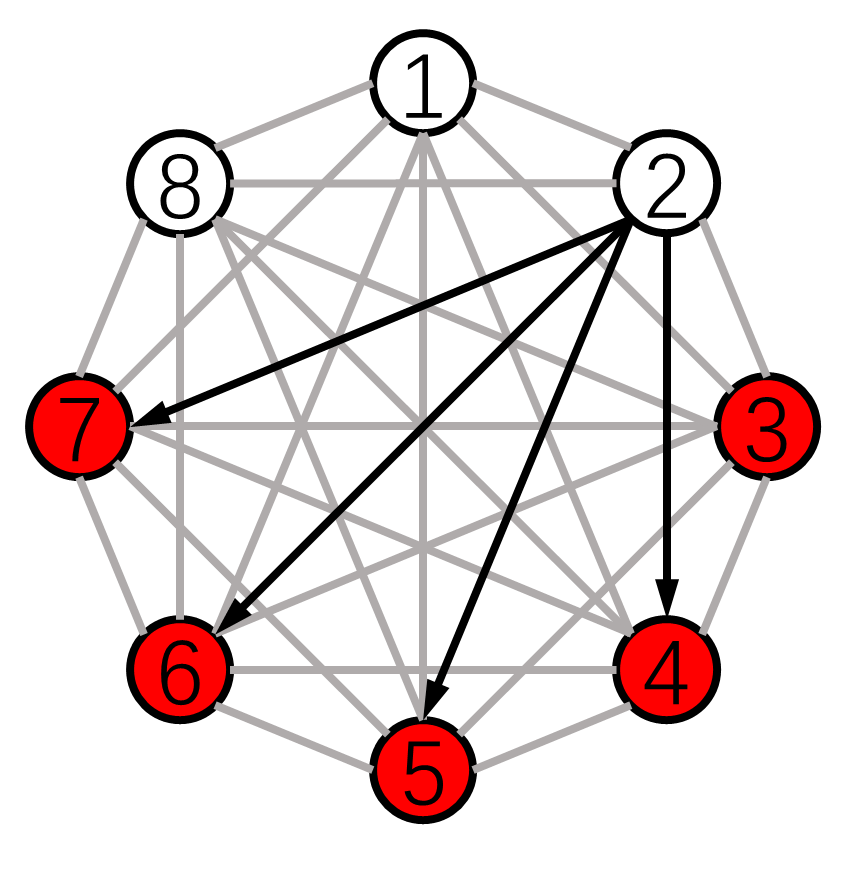}

\vspace{-3pt}

\caption{The 8-node network satisfying the condition for Algorithm~3 under the $1$-local model.}\label{8node}
\end{figure}

\begin{figure}[t]
\centering

\includegraphics[width=3.4in,height=1.4in]{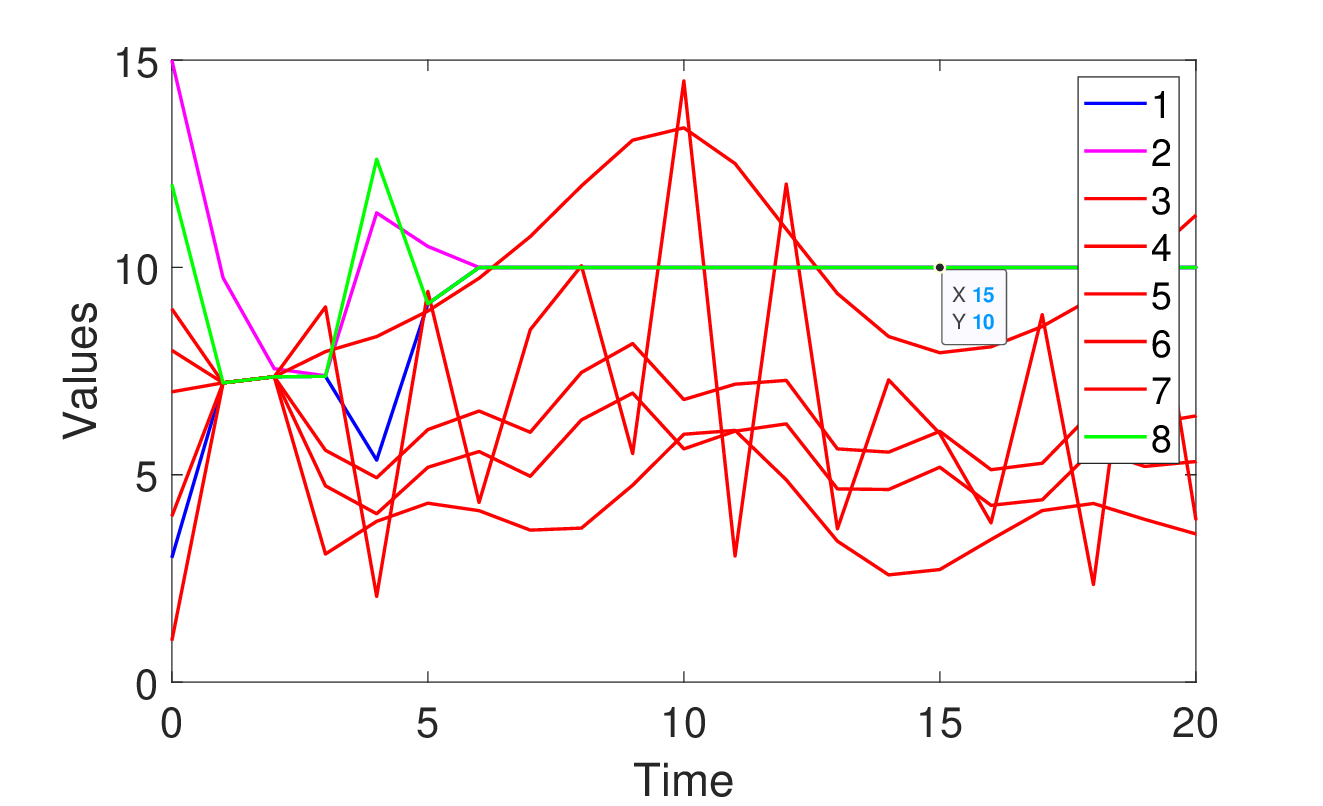}

\vspace{-3pt}
\caption{Time responses of Algorithm~3 in the 8-node network in Fig.~\ref{8node}.}
\label{A3-8node}
\end{figure}

\textit{3) Over Half of the Nodes Turn Malicious:} We conduct the simulation of Algorithm~3 under an extremely adversarial case, where over half of the nodes in the network turn malicious. Consider the 8-node network in Fig.~\ref{8node} with $\mathcal{A}=\{3, 4, 5, 6, 7\}$. It is almost complete except that there are 4 directed edges from node 2. Moreover, it satisfies the condition for Algorithm~3 under the 1-local model for non-full access node 2.

Set the initial values $x[0]=[3\ 15\ 9\ 8\ 4\ 7\ 1\ 12]^T$. The simulation result is presented in Fig.~\ref{A3-8node}. All malicious nodes simultaneously launch attacks at time $k=3$ by manipulating their values arbitrarily. However, these attacks are soon detected.
In Fig.~\ref{A3-8node}, the normal nodes eventually reach the average of their initial values $\overline{X}_{\mathcal{N}}=  10$. Therefore, we can conclude that resilient average consensus is still guaranteed using Algorithm~3 despite the erroneous effects from 5 malicious nodes.

\subsection{Simulations with Large-scale Undirected Networks}\label{sec_sim_undirected}

\textit{1) Simulation of Algorithm~2:}
Here, we show the effectiveness of Algorithm~2 by conducting a simulation in the 5-node undirected network in Fig.~\ref{small_graph}(a) with initial states $x[0]=[8\ 6\ 1\ 3\ 9]^T$.
It is 3-connected and with at least one common neighbor for every pair of neighbors.
Given these properties, Proposition~\ref{theorem_share} indicates that resilient average consensus can be achieved using Algorithm~2 under the 2-total malicious model. 
Let the adversary set $\mathcal{A}=\{4, 5\}$.
The simulation result under attacks is shown in Fig.~\ref{A2-5node}. Malicious node 5 first attacks other agents by transmitting arbitrary values at time $k=4$ and it is soon detected by its normal neighbors. At time $k=13$, nodes 1, 2, 3 and 4 almost reach the average of their initial values (i.e., 4.5). However, node 4 starts to manipulate its own value. Such an attack is also quickly detected. Then the normal nodes reach the average of their initial values $\overline{X}_{\mathcal{N}}= 5$.

\textit{2) Simulation of Algorithm~3:} In this part, we carry out the simulation of Algorithm~3 in a large-scale network, which is constructed by the method proposed in Section~\ref{construction}. Specifically, we consider the 30-node network in Fig.~\ref{10layer}. It has a 10-layer structure satisfying the condition for Algorithm~3 under the 1-local model. The malicious nodes are indicated in red with $\mathcal{A}=\{3, 6, 15, 18, 27, 30\}$.

\begin{figure}[t]
\centering

\includegraphics[width=3.4in,height=1.4in]{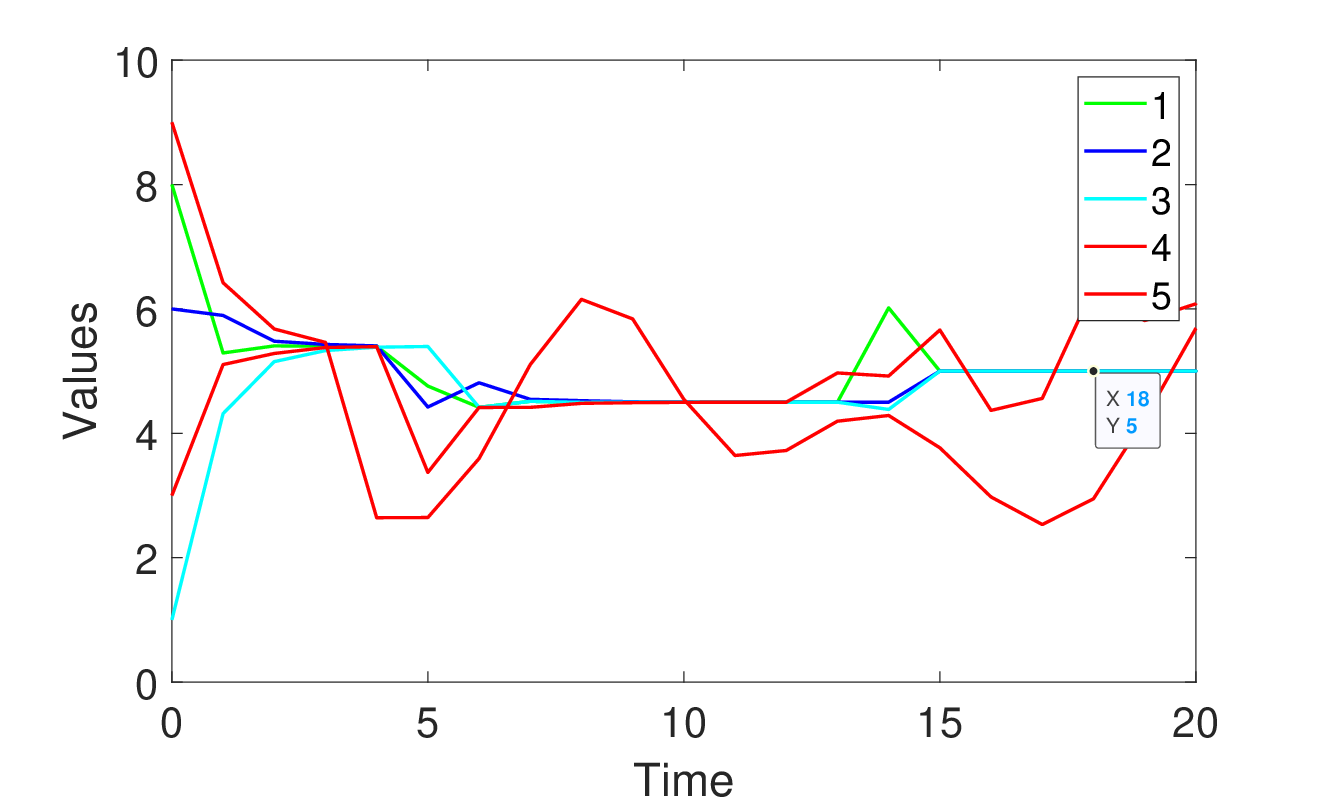}

\vspace{-3pt}
\caption{Time responses of Algorithm~2 in the 5-node network in Fig.~\ref{small_graph}(a).}
\label{A2-5node}
\end{figure}

\begin{figure}[t]
\centering

\includegraphics[width = 0.4\textwidth ]{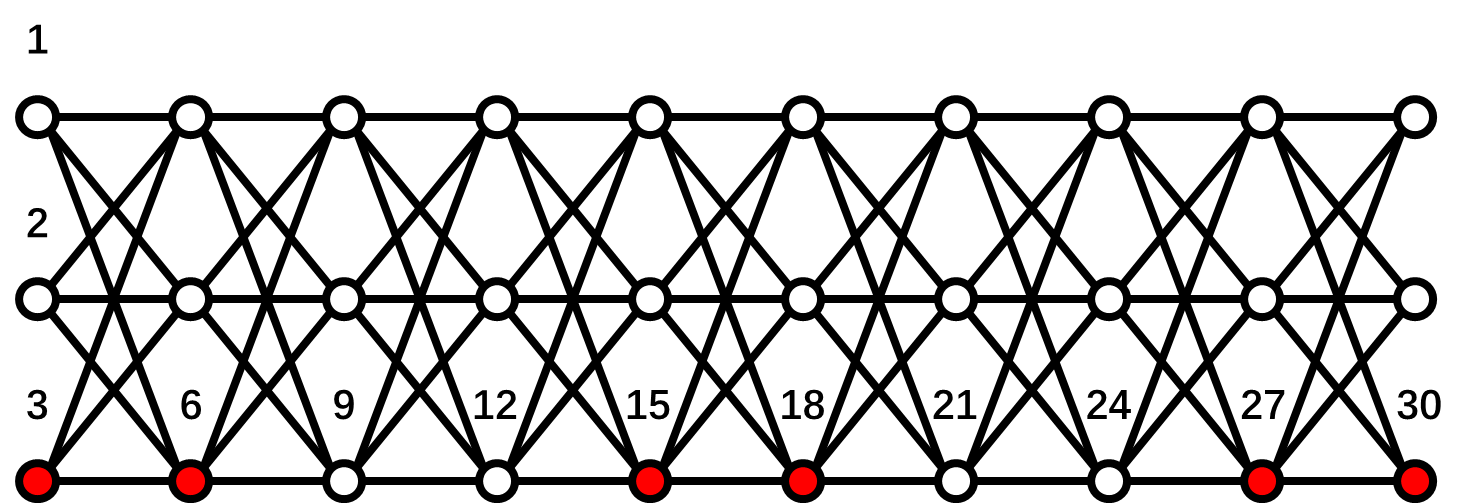}

\vspace{-3pt}
\caption{A large-scale network satisfying the condition for Algorithm~3 under the 1-local model.}\label{10layer}
\end{figure}

Let the initial values $x[0]=[8\ 7\ 5\ 3\ 2\ 11\  1\ 4\  6\ 9\ $ $ 10\ 12\ 11\ 13\  14\ 3\ 5\ 2\ 8\ 7\ 5\ 3\ 2\ 11\ 1\ 4\ 6\ 9\ 10\ 12]^T$.
The simulation results of Algorithm~3 without and with attacks are shown in Figs.~\ref{A3-10layer}(a) and (b), respectively.
One can see in Fig.~\ref{A3-10layer}(a) that all nodes achieve average consensus $\overline{X}= 6.8$ using Algorithm~3 although the convergence is slow due to the large network size.
As for the results of nodes under attacks, it shows in Fig.~\ref{A3-10layer}(b) that at time $k=9$, all the 6 malicious nodes start to manipulate their values through cooperating with their malicious neighbors and changing the past values of each other in their information sets. Such attacks are soon detected by their normal neighbors. Thereafter, the normal nodes reach the average of their initial values $\overline{X}_{\mathcal{N}}=  6.4166$. The RAC problem is solved by Algorithm~3 in the presence of 6 malicious nodes. Note that Algorithm~3 can still guarantee resilient average consensus if any one of the nodes become malicious in each one of the 6 layers containing malicious nodes currently. This is because the malicious nodes also satisfy the 1-local model in this case. We finally emphasize that none of the methods in \cite{zheng2021accurate,hadjicostis2023trustworthy,hadjicostis2023identification} can handle the above case of neighboring malicious nodes.

\begin{figure}[t]
\centering

\subfigure[\scriptsize{No attack.}]{
	\includegraphics[width=3.4in,height=1.4in]{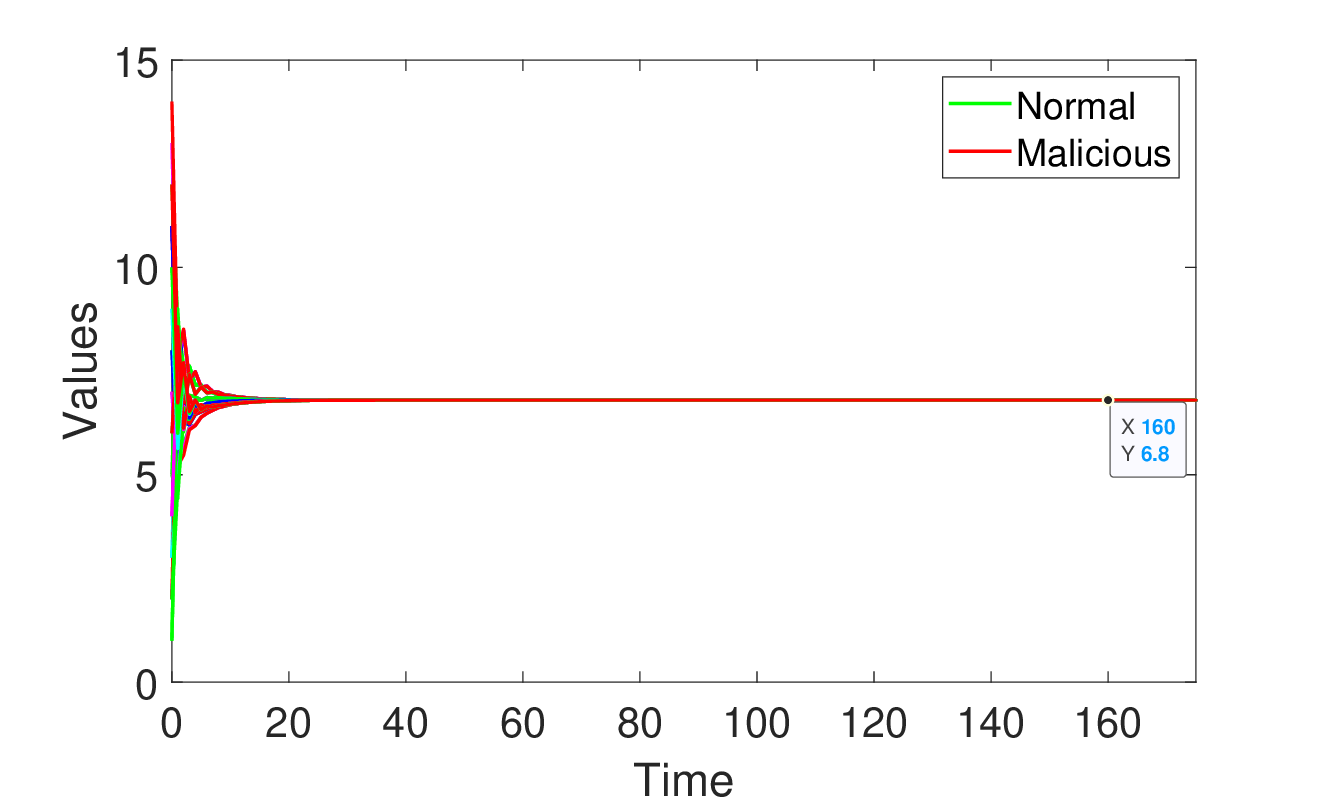}
}

\vspace{-3pt}

\subfigure[\scriptsize{Under attacks.}]{
	\includegraphics[width=3.4in,height=1.4in]{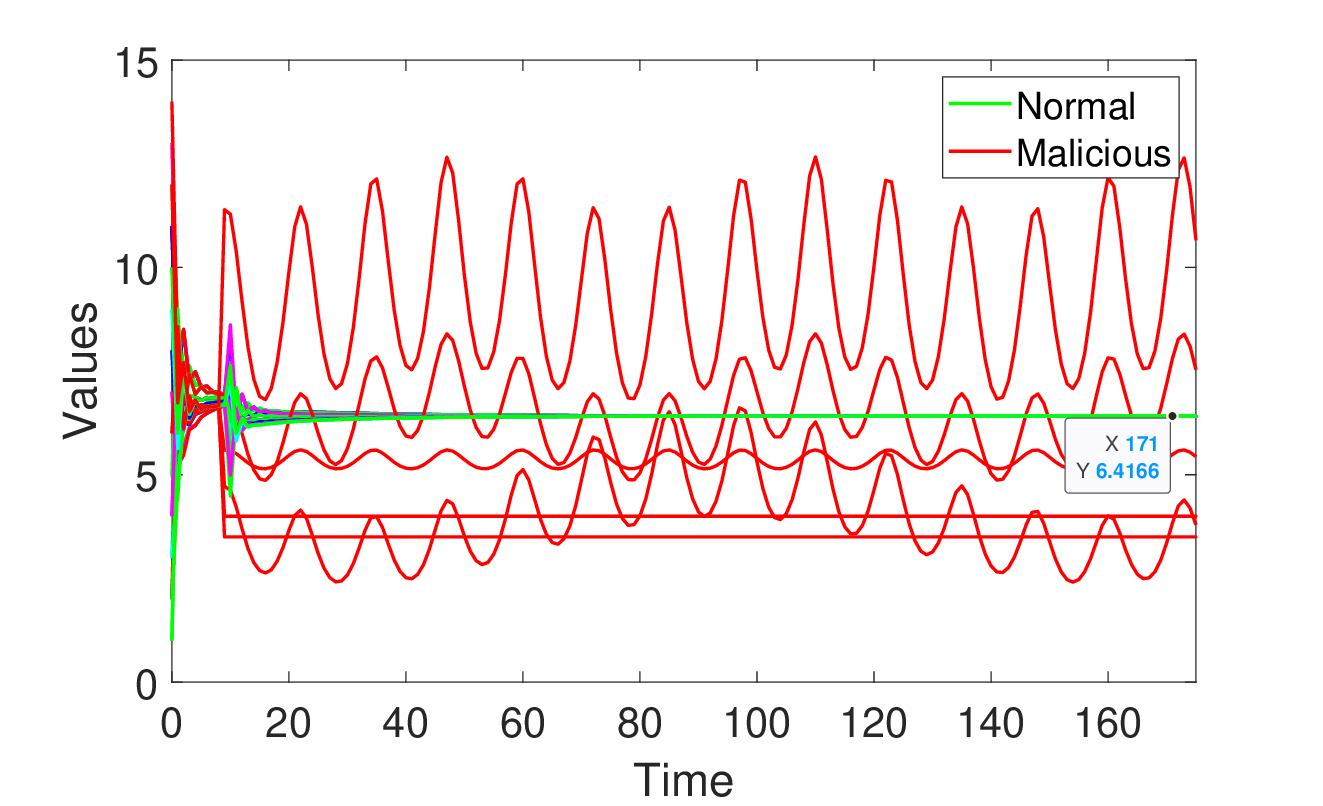}
}

\vspace{-3pt}
\caption{Time responses of Algorithm~3 in the 10-layer network in Fig.~\ref{10layer}.}
\label{A3-10layer}
\end{figure}

\section{Conclusion}
In this paper, we have investigated the problem of resilient average consensus in the presence of adversaries.
We have proposed a distributed iterative detection and averaging algorithm for normal agents to achieve resilient average consensus in general directed topologies. For the detection part, we have proposed two distributed algorithms and the second one can achieve fully distributed detection of malicious agents. For the averaging part, it can precisely preserve the sum of the initial values of normal agents.
Moreover, we have fully characterized the network requirement for the algorithms to successfully achieve resilient average consensus.
Compared to the existing detection approaches, our method is the only one that can handle the case of neighboring malicious agents. Besides, we have solved the resilient average consensus in directed networks, whereas the existing detection approaches studied undirected networks only. Moreover, in comparison with the existing secure broadcast and retrieval approach \cite{dibaji2019resilient}, our algorithm can save storage as each agent keeps only the values of two-hop neighbors.
In the end, we have provided extensive numerical examples to show the effectiveness of the proposed algorithms.

In future works, we are interested in applying our algorithms to various applications of average consensus where security needs to be enhanced, e.g., the economic dispatch problem and the PageRank problem.

\appendices

\section*{Appendix}

\section*{Proof of Lemma \ref{directed_minimum_degree}}

\begin{proof}
The proof is shown in two stages. First, we show that the clique structure (see the examples in Fig.~\ref{proof_indegree}(a)) is the minimum subgraph not having any node with in-degree more than $2f$ while satisfying the condition for Algorithm~3.
Due to the $f$-local model, each node must have at least $f$ in-neighbors. It is obvious that if any node uses the majority voting structure (i.e., $2f+1$ two-hop paths) to obtain the original value of a two-hop in-neighbor or an out-neighbor, then such a node will have at least $2f+1$ in-neighbors. Consider node $i$ with $f$ in-neighbors. By the above discussion, it has undirected edges to the $f$ in-neighbors, which results in these in-neighbors being two-hop in-neighbors to each other. Thus, by the same argument, there must be undirected edges between them.
Therefore, the clique is the only structure satisfying the condition for Algorithm~3 while not having any node with in-degree more than $2f$.

Next, we show the minimum in-degree of the whole graph. Since the graph is strongly connected, there exist bi-directional edges (one undirected edge or two separate directed edges) connecting two subgraphs. For example, we take the undirected edge between nodes $i$ and $j$ in Fig.~\ref{proof_indegree}(b). Then other nodes in the right subgraph become two-hop neighbors of node $i$. By similar discussions as above, there exist undirected edges between node $i$ and all the neighbors of node $j$ (as indicated by the blue edges in the figure). As a result, node $i$ has $2f+1$ in-neighbors. Moreover, we can check that all the rest of nodes also have $2f+1$ in-neighbors to fulfill the condition for Algorithm~3. We conclude that the whole graph has the minimum in-degree no less than $2f+1$.
\end{proof}

\begin{figure}[t]
\centering


\subfigure[]{
	\includegraphics[width=3cm]{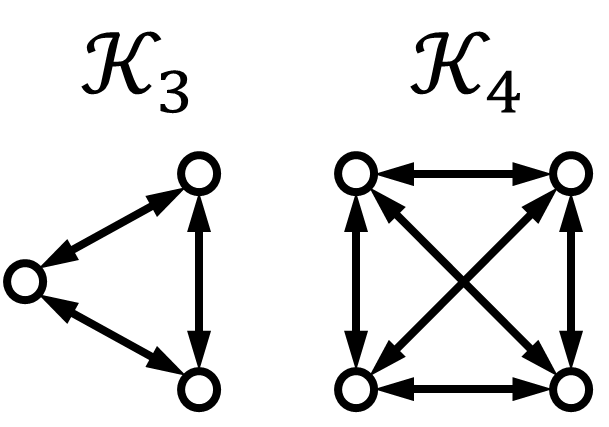}
	\label{9nodesforp2}
}
\quad
\subfigure[]{
	\includegraphics[width=3cm]{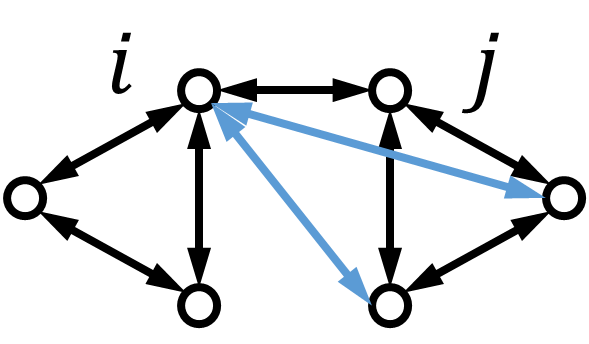}
	\label{connectcomplete}
}

\vspace{-3pt}
\caption{(a) Examples: A clique is a complete subgraph. (b) Illustration for the minimum in-degree.}
\label{proof_indegree}
\end{figure}

\section*{Proof of Proposition \ref{undirected_simple}}

\begin{proof}
We first show that for any node $i\in \mathcal{V}$, there must exist the minimum subgraph containing node $i$ as the one in Fig.~\ref{proof_connected}(a).
Recall from Lemma.~\ref{directed_minimum_degree} that an undirected graph satisfying the condition for Algorithm~3 has minimum in-degree no less than $2f+1$. In Fig.~\ref{proof_connected}(a), we set $f=1$ for illustration. The edges in blue and black represent the detectable path and the communication edge, respectively. Note that this subgraph also includes the middle nodes on the detectable path (not shown in the figure for convenience) if such path is not constructed by an undirected communication edge. It can be observed that in a minimum subgraph, after removing any node set being $f$-local, the remaining graph is connected. This means that there is at least one path connecting any two nodes in the remaining graph.

Now, consider any two minimum subgraphs with node sets $\mathcal{V}_1$ and $\mathcal{V}_2$ (see Fig.~\ref{proof_connected}(b)). There must be at least one edge between them since the whole graph is connected by assumption. There are three subcases for placing such an edge. These are between (i) $i$ and $j_1$, (ii) $i$ and $j$, (iii) $i_1$ and $j_1$. (Without loss of generality, select $j_1$ as one of $j$'s neighbors.) In cases (i) and (ii), node $j_1$ or $j$ becomes a direct neighbor of node $i$. Thus, node $j_1$ or $j$ is connected with any node in $\mathcal{V}_1$ after removing an $f$-local node set in $\mathcal{V}_1$. Since node $j_1$ or $j$ is also connected with any node in $\mathcal{V}_2$ after the removal, we can conclude that in cases (i) and (ii), any node in $\mathcal{V}_2$ is connected with any node in $\mathcal{V}_1$ after removing an $f$-local node set in the whole graph. 

In case (iii), nodes $i_1$ and $j_1$ become neighbors (see Fig.~\ref{proof_connected}(b)). There should be detectable paths between $i_1$ and $j$ and also between $i$ and $j_1$. If any of the two paths is constructed by an undirected communication edge, the result is the same as the one in case (i). So we examine the case where both paths are constructed by $2f+1$ two-hop communication paths. Suppose that node $i$ is connected to node $j_1$ through nodes $i_4$, $i_5$, and $i_6$. The three nodes become two-hop neighbors of node $j$ and there should be detectable paths to node $j$. In this case, even if we remove the $f$-local model consisting of both nodes $i_1$ and $j_1$, node $i$ and node $j$ are connected with each other, and so are the rest of the nodes in $\mathcal{V}_1$ and $\mathcal{V}_2$. Note that removing both nodes $i_1$ and $j_1$ does not violate the $f$-local model since they do not have common normal neighbors. Finally, since the malicious set is $f$-local, we conclude that the normal network induced by the normal agents in $\mathcal{N} \subseteq \mathcal{V}$ is connected.
\end{proof}

\begin{figure}[t]
\centering


\subfigure[]{
	\includegraphics[width=2cm]{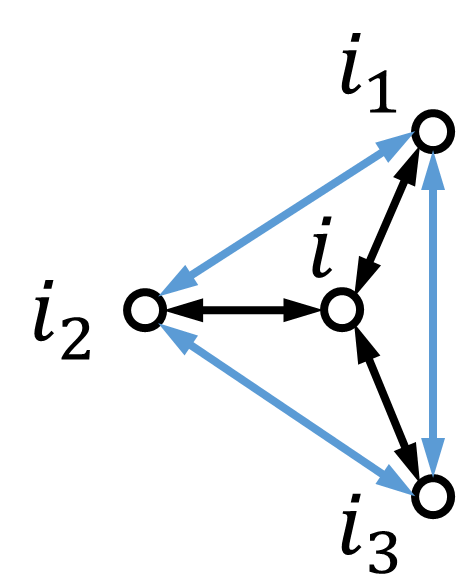}
}
\quad
\hspace{3pt}
\subfigure[]{
	\includegraphics[width=5cm]{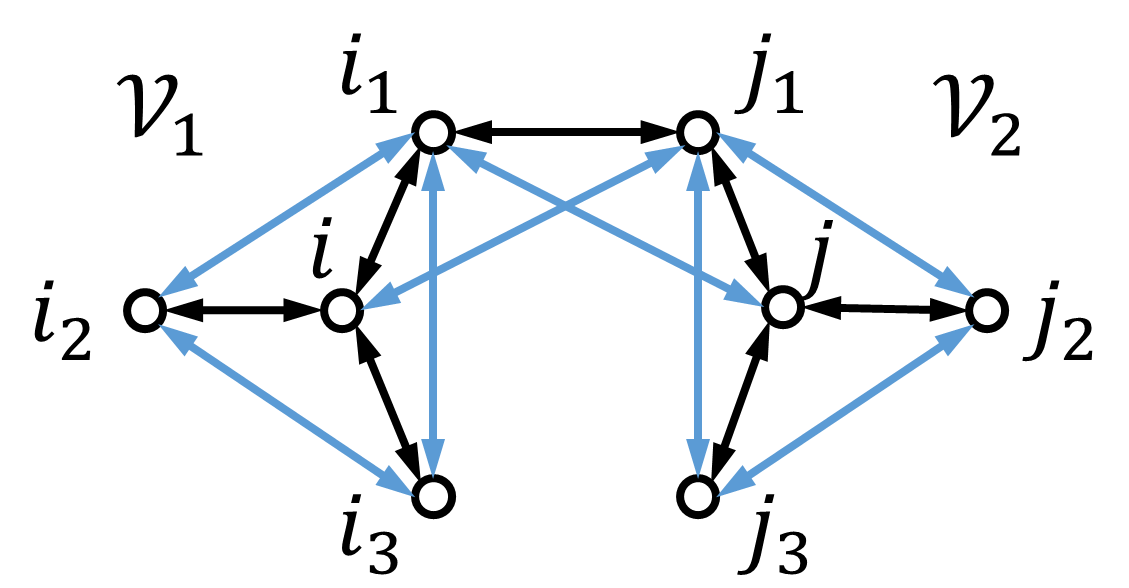}
}

\vspace{-3pt}
\caption{(a) The minimum subgraph of an undirected graph satisfying the condition for Algorithm~3. (b) Illustration for two connected subgraphs $\mathcal{V}_1$ and $\mathcal{V}_2$.}
\label{proof_connected}
\end{figure}

\vspace{-1cm}

\begin{IEEEbiography}[{\includegraphics[width=1in,height=1.25in,clip,keepaspectratio]{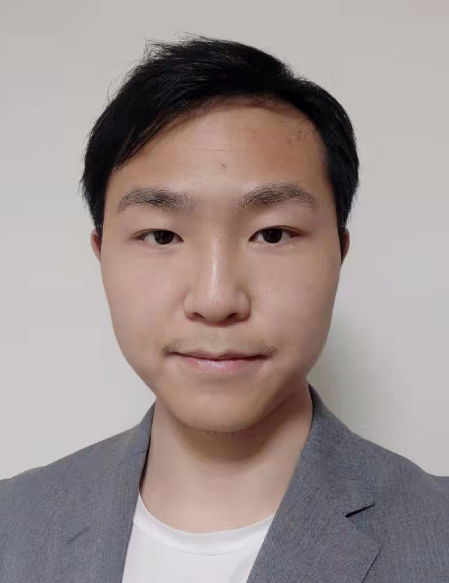}}]{Liwei Yuan} (Member) received the B.E. degree in Electrical Engineering
and Automation from Tsinghua University,
China, in 2017, and the Ph.D. degree in Computer
Science from Tokyo Institute of Technology, Japan, in
2022.
He is currently a Postdoctoral Researcher in the College 
of Electrical and Information Engineering at Hunan 
University, Changsha, China. His current
research focuses on security in multi-agent systems
and distributed algorithms.
\end{IEEEbiography}

\vspace{-1cm}

\begin{IEEEbiography}[{\includegraphics[width=1in,height=1.25in,clip,keepaspectratio]{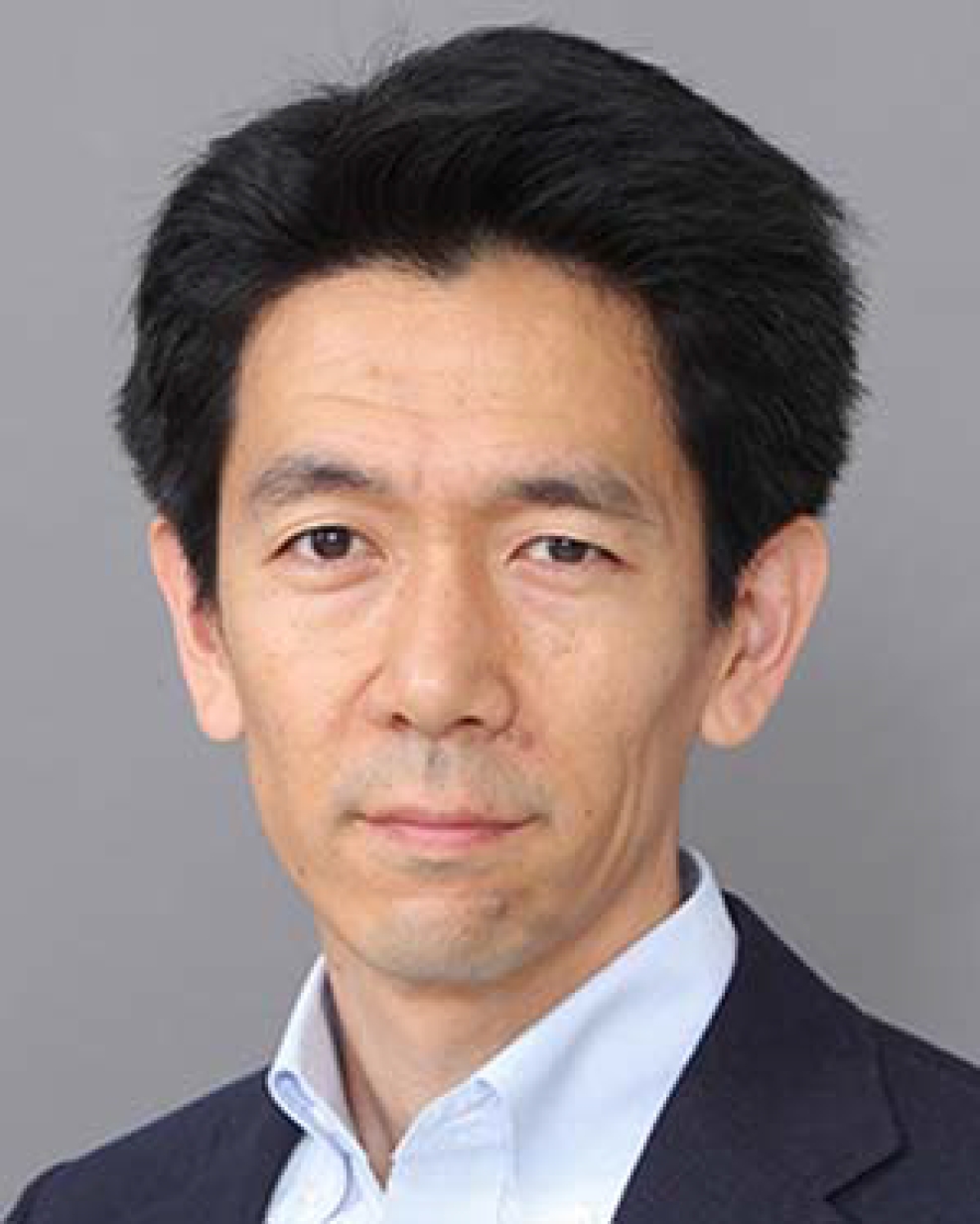}}]{Hideaki Ishii} (M'02-SM'12-F'21) received the
M.Eng.\ degree in applied systems science from
Kyoto University, Kyoto, Japan, in 1998, and the
Ph.D. degree in electrical and computer engineering
from the University of Toronto, Toronto,
ON, Canada, in 2002. He was a Postdoctoral Research
Associate at the University of Illinois at Urbana-Champaign,
Urbana, IL, USA, from 2001 to
2004, and a Research Associate at
The University of Tokyo, Tokyo, Japan, from 2004 to 2007.
He was an Associate Professor and Professor at the Department of Computer Science,
Tokyo Institute of Technology, Yokohama, Japan in 2007--2024. Currently, he is a Professor at the 
Department of Information Physics and Computing at The University of Tokyo, Tokyo, Japan.
He was a Humboldt Research Fellow at the University of Stuttgart
in 2014--2015. He has also held visiting positions at CNR-IEIIT at
the Politecnico di Torino, the Technical University of Berlin, and
the City University of Hong Kong. His research interests
include networked control systems, multiagent systems, distributed algorithms,
and cyber-security of control systems.

Dr.~Ishii has served as an Associate Editor for Automatica, 
the IEEE Control Systems Letters, the IEEE Transactions on Automatic Control, 
the IEEE Transactions on Control of Network Systems,
and the Mathematics of Control, Signals, and Systems.
He was a Vice President for the IEEE Control Systems Society (CSS) in 2022--2023.
He was the Chair of the IFAC Coordinating Committee on Systems and
Signals in 2017--2023.
He served as the IPC Chair for the IFAC World Congress 2023 held in Yokohama, Japan.
He received the IEEE Control Systems Magazine Outstanding Paper
Award in 2015. Dr.~Ishii is an IEEE Fellow.
\end{IEEEbiography}

\end{document}